\documentclass[a4paper,twocolumn,11pt,accepted=2023-03-10]{quantumarticle}
\pdfoutput=1
\usepackage[utf8]{inputenc}
\usepackage[english]{babel}
\usepackage[T1]{fontenc}
\usepackage{amsmath}
\usepackage{amsthm}
\usepackage{hyperref}
\usepackage[numbers,sort&compress]{natbib}
\usepackage{mathtools}
\usepackage{dsfont}
\usepackage{color}

\usepackage{tikz}
\usepackage{lipsum}
\usepackage{cleveref}

\newcommand{\C}{\mathcal{C}}

\newcommand{\N}{\mathcal{N}}
\renewcommand{\P}{\mathcal{P}}
\renewcommand{\S}{\mathcal{S}}
\newcommand{\E}{C_N}
\newcommand{\supervec}[1]{|#1\rangle\rangle}
\newcommand{\conjsupervec}[1]{\langle\langle #1|}
\newcommand{\inner}[2]{\langle \langle #1 | #2 \rangle \rangle}
\newcommand{\tr}{\text{Tr}}
\newcommand{\state}[1]{|#1\rangle}
\newcommand{\conjstate}[1]{\langle #1|}
\newcommand{\pure}[1]{\state{#1}\conjstate{#1}}
\newcommand{\1}{\mathds{1}}

\newtheorem{proposition}{Proposition}[section]

\newcommand{\affilITP}{Institute for Theoretical Physics, ETH Z\"{u}rich, CH-8093 Z\"urich, Switzerland.}
\newcommand{\affilKON}{Department of Physics, University of Konstanz, 78464 Konstanz, Germany.}

\begin{document}

\title{Efficient separation of quantum from classical correlations for mixed states with a fixed charge}

\author{Christian Carisch}\affiliation{\affilITP}\orcid{0000-0001-7393-614X}
\author{Oded Zilberberg}\affiliation{\affilKON}\orcid{0000-0002-1759-4920}

\maketitle

\begin{abstract}
Entanglement is the key resource for quantum technologies and is at the root of exciting many-body phenomena. However, quantifying the entanglement between two parts of  a real-world quantum system is challenging when it interacts with its environment, as the latter mixes cross-boundary classical with quantum correlations. Here, we efficiently quantify quantum correlations in such realistic open systems using the operator space entanglement spectrum of a mixed state. If the system possesses a fixed charge, we show that a subset of the spectral values encode coherence between different cross-boundary charge configurations. The sum over these values, which we call ``configuration coherence'', can be used as a quantifier for cross-boundary coherence. Crucially, we prove that for purity non-increasing maps, e.g., Lindblad-type evolutions with Hermitian jump operators, the configuration coherence is an entanglement measure. Moreover, it can be efficiently computed using a tensor network representation of the state's density matrix. We showcase the configuration coherence for spinless particles moving on a chain in presence of dephasing. Our approach can quantify coherence and entanglement in a broad range of systems and motivates efficient entanglement detection.
\end{abstract}

In quantum mechanics, particles can become far more correlated than classically possible. Such correlations, dubbed entanglement, are a key resource in the present-day quantum revolution. For example, entanglement is harvested in quantum information processing devices~\cite{nielsen_chuang_2010, boixo_et_al_2018, neill_et_al_2018, arute_et_al_2019}, error-correction schemes~\cite{bennett_et_al_1996, cory_et_al_1998, schindler_et_al_2011, andersen_et_al_2020, krinner_et_al_2021}, quantum detectors that break sensitivity limits~\cite{pezze_smerzi_2009, demkowicz_et_al_2012, zhou_et_al_2018},  or secure quantum communication protocols~\cite{long_et_al_2007, hu_et_al_2016, zhang_et_al_2017}.
Entanglement can also lead to unique effects such as teleportation~\cite{bouwmeester_et_al_1997, furusawa_et_al_1998, nielsen_et_al_1998, riebe_et_al_2004}, the formation of strong correlations in many-body systems~\cite{nozieres_blandin_1980, kondo_2012, basko_et_al_2006, nandkishore_huse_2015, stormer_et_al_1999, avella_mancini_2012, bruus_flensberg_2004, carusotto_ciuti_2013, bloch_et_al_2008, campagnano_et_al_2012, shapourian_et_al_2017, wolf_et_al_2019, wolf_et_al_2021, lado_zilberberg_2019, strkalj_et_al_2021, khedri_et_al_2021, ferguson_et_al_2020, ferguson_et_al_2021, li_et_al_2018, szyniszewski_et_al_2019, liu_et_al_2022}, and the high efficiency of light-harvesting processes~\cite{sarovar_et_al_2010, caruso_et_al_2010, ishizaki_fleming_2010}.

The premise of quantum mechanics relies on having a wavefunction description for particles moving in a closed system. The wavefunction entails probability amplitudes for the state to be in different locations in the Hilbert space of the system. Commonly, the Hilbert space is very large, and entanglement has
become a modern tool for compressing the required information needed to properly describe a quantum state~\cite{nieuwenburg_zilberberg_2018, stocker_et_al_2022}. For example, in tensor network representations of quantum states, the Hilbert space is truncated such that only entangled regions are kept~\cite{perez_et_al_2007, schollwock_2005, schollwock_2011}. As such, measures for quantifying entanglement (e.g., entanglement entropy~\cite{nielsen_chuang_2010,islam_et_al_2015}) were developed to assess the potential usefulness of quantum resources, as well as to compress their representation.

In reality, however, all quantum systems are open, i.e., they are coupled to an environment and become correlated with it. The direct result of such coupling is that the state of the system can become mixed, i.e., lose its  entanglement. This is best described by considering the system's density matrix, which is kin to a covariance matrix of the state's probability amplitudes. As the density matrix describes both classical and quantum correlations of mixed states, it is notoriously difficult to distill the amount of entanglement in the system. 
In fact, it is known to be NP-hard to decide whether a mixed state is entangled or not~\cite{gurvits_2003, gharibian_2008}.
Many mixed-state entanglement measures have been developed, e.g., (R\'enyi) negativity~\cite{vidal_werner_2002, calabrese_et_al_2012, calabrese_et_al_2013, wybo_et_al_2020, sang_et_al_2021}, squashed entanglement~\cite{christandl_winter_2004}, reflected entropy~\cite{dutta_faulkner_2021}, or number entanglement~\cite{ma_et_al_2022}.
However, so far none of them can be efficiently computed for reasonably large many-body systems.

In this work, we investigate the degeneracy structure of the operator space entanglement spectrum (OSES)~\cite{zanardi_2001, prosen_pizorn_2007, pizorn_prosen_2009} and use it to define a tensor network computable measure of cross-boundary coherence and entanglement.
We start by rigorously defining the OSES as the eigenvalues of the so-called $\C$-matrix~\cite{nieuwenburg_zilberberg_2018}, and analyse it for pure and mixed states separately.
\begin{enumerate}
    \item[I.] For pure states, we find that the $\C$-matrix is diagonal with respect to the state's Schmidt basis.
We show that the sum over degenerate OSES values quantifies quantum correlations and is closely related to the negativity.
\item[II.] Moving to mixed states, we rely on a fixed charge symmetry, e.g., fixed particle number.
We show that under this added assumption,  the $\C$-matrix is block diagonal. 
We identify that certain blocks have the same eigenvalues. 
Interestingly, the resulting degenerate OSES values reflect the coherence between different local charge configurations.
Hence, we propose their sum as a convex coherence quantifier, which we call the ``configuration coherence''.
Furthermore, we prove that for purity non-increasing charge conserving maps, the configuration coherence is an entanglement measure.
\end{enumerate} 
Importantly, due to easy access to the OSES in tensor network simulations of mixed states, the configuration coherence can be efficiently calculated for many-body system sizes beyond the reach of existing entanglement measures.  
Finally, we showcase the configuration coherence by measuring the coherence and entanglement in an open quantum system during Lindblad evolution, thus motivating its broad applicability.

The \textit{entanglement spectrum} (ES) of a pure quantum state $\state{\psi}$ is defined relative to a biparition (cut) of the system into two parts $A$ and $B$ [see Fig.~\ref{fig: Systems and Failure of OSES}(a)] as the spectrum of the reduced state $\rho_A=\tr_B\left(\pure{\psi}\right)$. Concurrently, the Schmidt decomposition of the state relative to this cut is
\begin{equation}
    \label{eq: pure state Schmidt decomposition}
    \state{\psi} = \sum_{i=1}^r \sqrt{\lambda_i} \state{i, \mu_i}\,,
\end{equation}
where $r\geq 1$ is the Schmidt rank, $\sqrt{\lambda_i} \geq 0$ are real-valued Schmidt values, and $\state{i,\mu_i} = \state{i}_A \otimes \state{\mu_i}_B$ with suitable orthonormal sets of states for systems $A$ and $B$. 
The ES of a state~\eqref{eq: pure state Schmidt decomposition} is given by the squares $\lambda_i$ of its Schmidt values~\cite{li_haldane_2008}. The corresponding von Neumann entropy,  $\S_\text{vN} \equiv -\tr\left( \rho_A \log \rho_A\right)=-\sum_i \lambda_i \log \lambda_i$, serves as an entanglement measure for pure states.

Similarly to the pure case, we can define the OSES of a density matrix $\rho$ relative to a biparition of the system into two parts $A$ and $B$. The density matrix can be written relative to this biparition as
\begin{equation}
    \label{eq: density matrix}
    \rho = \sum_{\substack{i,j \in A \\ \mu, \nu \in B}} \rho_{i,\mu;j, \nu} \state{ i, \mu}\conjstate{ j, \nu},
\end{equation}
where $\state{i, \mu} = \state{i}_A \otimes \state{\mu}_B$, and $\state{i}_A$ and $\state{\mu}_B$ are basis states of the two parts of the system. The density matrix is Hermitian, and hence the prefactors in Eq.~\eqref{eq: density matrix} follow the relation $\rho_{i,\mu;j, \nu}=\rho_{j, \nu;i, \mu}^*$. 
We define the vectorized density matrix as
\begin{equation}
	\label{eq: vectorized density matrix}
	\supervec{\rho} \equiv \sum_{\substack{i,j \in A \\ \mu, \nu \in B}} \rho_{i,\mu;j, \nu} \supervec{i,\mu;j,\nu}\,,
\end{equation}
which is obtained by stacking the columns of the density matrix~\eqref{eq: density matrix} into a column vector. The inner product over such column vectors in terms of their respective density operators is defined as $\inner{\rho}{\sigma}\equiv \tr(\rho^\dagger \sigma)$.
The OSES of $\rho$ consists of the eigenvalues of the matrix~\cite{nieuwenburg_zilberberg_2018}
\begin{align}
    \label{eq: configurational matrix}
    \C &= \tr_{B}\left( \supervec{\rho}\conjsupervec{\rho}\right) \nonumber \\ &= \sum_{\substack{i,j,k,l \in A \\ \mu, \nu \in B}}^{\phantom{*}} \rho_{i, \mu ; j, \nu} \rho_{k, \mu; l, \nu}^* \supervec{i,j}\conjsupervec{k,l}\,,
\end{align}
where $\tr_{B} \mathcal{O} = \sum_{\mu, \nu \in B} \conjsupervec{\mu,\nu} \mathcal{O} \supervec{\mu,\nu}$ is the partial trace over subsystem $B$. The matrix~\eqref{eq: configurational matrix} is a positive operator that involves correlations up to fourth order in the state's probability amplitudes, and we dub it the kurtosis matrix. As we show below, the OSES contains values encoding both classical correlations as well as entanglement, raising doubts concerning its naming convention. Interestingly, the sum over both classical and quantum OSES values equals to the purity of the system, i.e., $\tr \C = \tr\left( \supervec{\rho}\conjsupervec{\rho}\right)= \tr \rho ^2 \equiv \P(\rho)$. 

\begin{figure}[t!]
    \centering
    \includegraphics[width=\columnwidth]{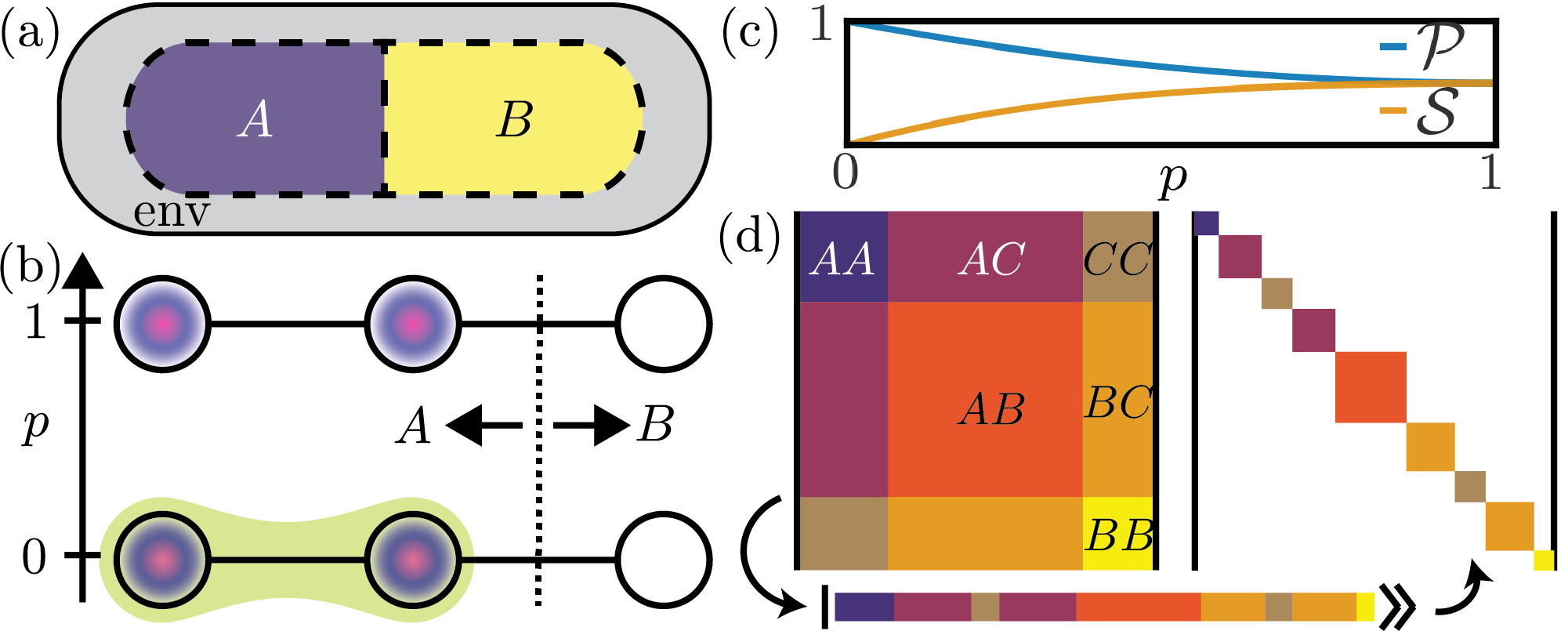}
    \caption{(a) An open bipartite system $A\cup B$ in contact with an environment. (b) A single particle on a 3-site system is distributed over sites 1 and 2 (subsystem $A$), as a function of mixing [cf.~Eq.~\eqref{eq: mix of pure and classical mixture}] from pure ($p=0$) to fully mixed ($p=1$). Yellow shading marks quantum correlations, which are lost in the fully mixed limit. (c) Purity $\P$ of the whole system and OSEE $\S$ between $A$ and $B$ in (b) as a function of mixing $p$. The OSEE increases despite the fact that the mixing does not affect cross-boundary correlations.
    (d) Pictorial derivation of the block diagonal structure of the $\C$-matrix [cf.~Eq.~\eqref{eq: configurational matrix for n particles}]. For example, the 2-particle density matrix has entries corresponding to the particles' configurations, e.g., both in subsystem $A$ ($AA$), or one in $A$ and one in $B$ ($AB$). When a particle is coherently delocalized across the bipartition, we use $C$. After vectorizing the density matrix [cf.~Eq.~\eqref{eq: vectorized density matrix}], the partial trace reveals the block diagonal structure [cf.~Eqs.~\eqref{eq: configurational matrix} and~\eqref{eq: configurational matrix for n particles}].}
    \label{fig: Systems and Failure of OSES}
\end{figure}

Ostensibly, we would like to employ the operator space entanglement entropy (OSEE), $\S\equiv~-\tr\left(\C\log\C\right)$, as an entanglement measure. Yet, it appears to be sensitive to both classical and quantum correlations~\cite{prosen_pizorn_2007, dubail_2017}. The latter is evident from our construction~\eqref{eq: configurational matrix} using a counterexample: consider a pure state with a single excitation residing solely within subsystem $A$, e.g., subsystem $A$ is composed of states $\state{1}$ and $\state{2}$, whereas subsystem $B$ of state $\state{3}$, see Fig.~\ref{fig: Systems and Failure of OSES}(b). We take the quantum state to be in an equal superposition, $\state{\psi}=(\state{1} + \state{2})/\sqrt{2}$. The corresponding OSES has a single nonvanishing value $\Lambda_A$ [cf.~Eq.~\eqref{eq: configurational matrix for n particles} and discussion below]. Hence, for our pure system $\Lambda_A=\tr(\C)=\P(\rho)=1$. Correspondingly, $\S \equiv -\Lambda_A \log \Lambda_A = 0$ as expected for a pure product state. We now locally couple subsystem $A$ to a dephasing environment, i.e., no particles leak out, but the system decoheres into a mixed state $\rho'$ after some time. As the particle remains in subsystem $A$, we still have a single eigenvalue $\Lambda_A'$ that corresponds to a reduced purity $\P(\rho')<1$ of the system. Thus, we obtain that the OSEE increases to $\S'=-\Lambda_A' \log \Lambda_A' > 0$ even though the local operation on subsystem $A$ cannot have generated entanglement between subsystems $A$ and $B$. This is a first important observation of this work.

In Fig.~\ref{fig: Systems and Failure of OSES}(c), we show the outcome of our counterexample with increasing dephasing. The latter is obtained by mixing the pure state with the classical mixture of the particle being either in state $\state{1}$ or $\state{2}$, namely by
\begin{equation}
	\label{eq: mix of pure and classical mixture}
	\rho_p = (1-p) \pure{\psi} + p\sigma\,,
\end{equation}
with $\sigma=(\pure{1} + \pure{2})/2$, see Fig.~\ref{fig: Systems and Failure of OSES}(b). The OSEE increases with increasing weight $p$ of the separable classical mixture, confirming the deficiency of the OSEE as an entanglement measure for mixed states. Crucially, we identify that part of the OSES bears no entanglement information, e.g., $\Lambda_A$ in our example. Hence, any entanglement measure should filter out such values, which may be challenging for a many-body system on a large Hilbert space. 

For pure states $\state{\psi}$, however, the filtering is relatively straightforward: we can write the  $\C$-matrix~\eqref{eq: configurational matrix} of a density matrix $\pure{\psi}$ using the state's Schmidt basis [cf.~Eq.~\eqref{eq: pure state Schmidt decomposition}] as
\begin{equation}
    \label{eq: C matrix pure}
    \C = \sum_i \lambda_i^2 \supervec{i,i}\conjsupervec{i,i} + \sum_{j\neq i} \lambda_i \lambda_j \supervec{i,j}\conjsupervec{i,j}\,.
\end{equation}
In this basis, the $\C$-matrix is diagonal and its spectrum consists of $r$ eigenvalues of type $\lambda_i^2$ and $r(r-1)/2$ two-fold degenerate eigenvalues of type $\lambda_i\lambda_j$. Thus, in this pure limit, the OSES of the density matrix is equivalent to the outer product of the ES of the state~\cite{nieuwenburg_huber_2014, nieuwenburg_zilberberg_2018}. We can also readily verify using Eq.~\eqref{eq: C matrix pure} that only a single nonvanishing $\lambda_i^2$ value appears in the pure limit of the example of Fig.~\ref{fig: Systems and Failure of OSES}(b).

Now, recall that a pure state~\eqref{eq: pure state Schmidt decomposition} is entangled if and only if its Schmidt rank is $r>1$.
As the second sum in Eq.~\eqref{eq: C matrix pure} vanishes for $r=1$ and is finite and positive for $r>1$, we propose the sum over these eigenvalues as a quantifier of quantum correlations in pure states,
\begin{equation}
    \label{eq: pure state entanglement measure}
    \E (\state{\psi}) \coloneqq \sum_{j\neq i} \lambda_i \lambda_j\,.
\end{equation}
In other words, we obtain the quantum correlations of a pure state as the sum over inherently degenerate eigenvalues of the matrix $\C$ and filter out the $\lambda_i^2$ values.

Interestingly, the quantity~\eqref{eq: pure state entanglement measure} is closely related to the negativity of the state, which is defined as the absolute value of the sum over all negative eigenvalues of the partial transpose $\rho^{T_B}$ of the density matrix~\cite{vidal_werner_2002}. Indeed, using the Schmidt decomposition~\eqref{eq: pure state Schmidt decomposition}, the negativity reads 
\begin{equation}
\N~=~\frac{1}{2}\sum_{j\neq i}\sqrt{\lambda_i}\sqrt{\lambda_j}\,.
\end{equation}
Thus, our Eq.~\eqref{eq: pure state entanglement measure} defines a new way to calculate the negativity for pure states: it is obtained as the sum over the square roots of the degenerate eigenvalues of the $\C$-matrix. This is a second important observation of this work. Furthermore, the definition of the quantity~\eqref{eq: pure state entanglement measure} via the matrix $\C$ lends a natural extension to open systems. The remaining challenge involves the identification of $\C$-matrix eigenvalues that encode cross-boundary quantum correlations for mixed states.

\begin{figure}[t!]
    \centering
    \includegraphics[width=\columnwidth]{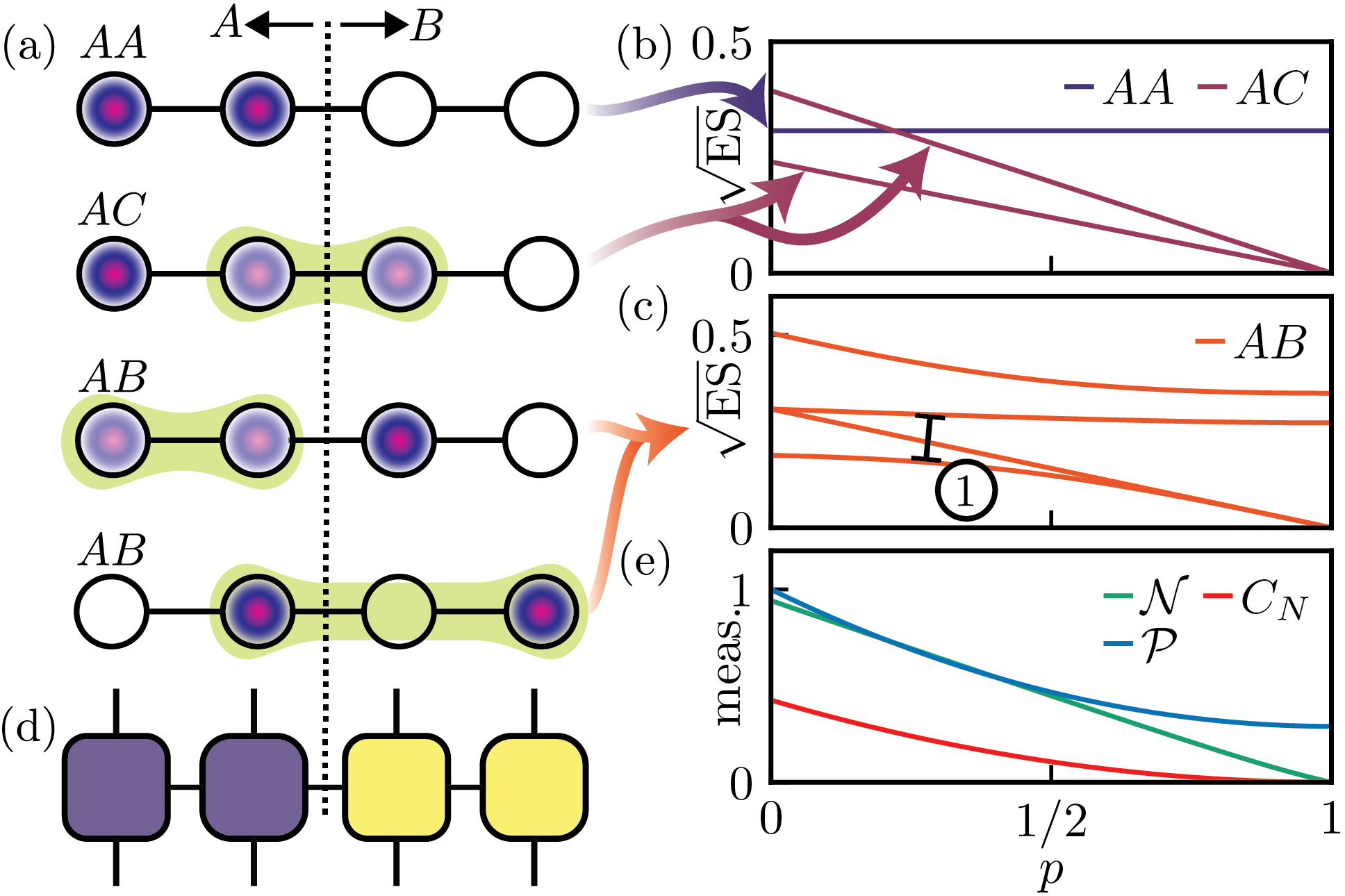}
    \caption{(a) Examples of configurations of two particles on four sites with respect to a bipartition in the middle. 
    Yellow shadings mark quantum correlations. (b) and (c) OSES corresponding to the configurations in (a) along a mixing interpolation [cf.~Eq.~\eqref{eq: mix of pure and classical mixture}]. (1) marks an avoided crossing between spectral values of a classically and a quantum correlated $AB$ configuration.
    (d) MPDO representation of a density matrix on 4 sites.
    (e) Negativity $\N$, purity $\P$, and configuration coherence $\E$ for the spectrum in (b) and (c).}
    \label{fig: 2 particle example}
\end{figure}

Before turning to discuss the OSES of mixed states in more detail, we present some general comments about mixed state entanglement.
As mentioned above, the mixed state separability problem is NP-hard~\cite{gurvits_2003, gharibian_2008}.
The complexity of the problem should therefore be reduced in order to develop a computable mixed state entanglement measure.
Here, we impose a symmetry to restrict the relevant Hilbert space dimension.
Within this space, we consider states which have a fixed value of the symmetry's conserved charge.
Without loss of generality, we study systems with a fixed number of $N$ particles.
This restriction leads to degeneracies in the OSES, which allow us to define the configuration coherence as a measure of quantum coherence.
We prove that the configuration coherence is an entanglement measure under purity non-increasing charge conserving maps.
Our approach is related to recent studies of symmetry-resolved coherence and entanglement measures~\cite{cornfeld_et_al_2018, macieszczak_et_al_2019, ma_et_al_2022}.

We write the density matrix~\eqref{eq: density matrix} using the basis states $\state{i_{n},\mu_{n}}$ with $0 \leq n \leq N$ particles in subsystem $B$ and $N-n$ particles in subsystem $A$. The $\C$-matrix in this basis is block diagonal,
\begin{align}
    \label{eq: configurational matrix for n particles}
    \C &=  \bigoplus_{n, n'=0}^N \C_{n,n'}\,,
\end{align}
with blocks 
\begin{equation}
\C_{n,n'}=\sum_{\substack{i_{n},j_{n'}  k_n,l_{n'}}} c^{\phantom{\dagger}}_{i_n,j_{n'};k_n,l_{n'}} \supervec{i_n,j_{n'}}\conjsupervec{k_n,l_{n'}}
\end{equation}
and coefficients
\begin{equation}
c^{\phantom{\dagger}}_{i_n,j_{n'};k_n,l_{n'}} = \sum_{\mu_n,\nu_{n'}}\rho_{i_n,\mu_{n};j_{n'},\nu_{n'}} \rho_{k_n,\mu_n;l_{n'},\nu_{n'}}^*\,.
\end{equation}
A graphical derivation of this block diagonal form is shown in Fig.~\ref{fig: Systems and Failure of OSES}(d). Importantly, we can interpret the blocks in terms of configurations of the $N$ particles with respect to the bipartition: the block $\C_{n,n'}$ contains all information on $\text{min}(n,n')$ particles that are fully in subsystem $B$, $N-\text{max}(n,n')$ particles that are fully in subsystem $A$, and $\text{max}(n,n')-\text{min}(n,n')$ particles that are coherently distributed across the cut. 

In the following, we compare the eigenvalues of the blocks $\C_{n, n'}$ with the pure case limit~\eqref{eq: C matrix pure} to distinguish classical from quantum correlations for mixed states. We accompany our discussion with an example of a chain with $N=2$ spinless particles residing on 4 sites, see Fig.~\ref{fig: 2 particle example}.
We begin with the blocks $\C_{0, 0}$ and $\C_{N,N}$, which are rank 1 and have eigenvalues $\Lambda_{0,0}~=~~\sum_{i,j} |\rho_{i,0;j,0}|^2$ and $\Lambda_{N,N}~=\sum_{\mu,\nu} |\rho_{0,\mu;0,\nu}|^2$, corresponding respectively to the scenario where all $N$ particles reside solely in subsystem $A$ or $B$, see Fig.~\ref{fig: 2 particle example}(a). These eigenvalues do not contain any information about cross-boundary coherence nor contribute to quantum correlations between subsystems $A$ and $B$. Indeed, these values are generally non-degenerate, and we identify that they reduce to eigenvalues of type $\lambda_i^2$ in the pure case limit, cf.~Eq.~\eqref{eq: C matrix pure}. Furthermore, such values do not vanish for a fully mixed state [see Fig.~\ref{fig: 2 particle example}(b)], justifying our choice to not include them in our correlation measure~\eqref{eq: pure state entanglement measure}.

The blocks $\C_{n,n'}$ for $n \neq n'$ describe a scenario where at least one of the particles is in a coherent cross-boundary state and clearly encode cross-boundary coherence. The blocks $\C_{n,n'}$ and $\C_{n',n}$ generate the same eigenvalues as they are related via a unitary transformation, $\C_{n,n'}=U\C_{n',n}U^{-1}$, with the permutation $U\supervec{i,j'} = \supervec{j',i}$. Hence, such coherent eigenvalues of $\C$ are inherently degenerate, and must map to the eigenvalues of type $\lambda_i \lambda_j$ in the pure case limit~\eqref{eq: C matrix pure}. Therefore, we employ these eigenvalues to extend the pure state correlation measure~\eqref{eq: pure state entanglement measure} to the mixed case. 
Specifically, we define the \emph{configuration coherence} as their sum,
\begin{align}
    \label{eq: configuration entanglement}
    \E (\rho) &\coloneqq \sum_{n\neq n'}\tr\left(\C_{n, n'}\right) \nonumber \\  &=
    \sum_{n\neq n'}\sum_{\substack{i_n,j_{n'} \\ \mu_n,\nu_{n'}}}|\rho_{i_n,\mu_n;j_{n'},\nu_{n'}}|^2\,.
\end{align}
Crucially, the configuration coherence contains only off-diagonal elements of the density matrix, which vanish in the fully decohered case, in conjunction with the fact that the fully decohered state contains no quantum correlations, see Figs.~\ref{fig: 2 particle example}(a), (b), and (e).
That the configuration coherence can be written in terms of off-diagonal density matrix entries explains our naming choice: quantum coherence is identified with the off-diagonal elements of the density matrix.

Before we turn to discuss the properties of the configuration coherence, let us highlight the key aspects of our derivation of Eq.~\eqref{eq: configuration entanglement}.
We started with the $\C$-matrix~\eqref{eq: configurational matrix for n particles} of a mixed state with a fixed particle number.
In any bipartite basis $\supervec{i_n, \mu_n}$, the $\C$-matrix is block-diagonal and the blocks $\C_{n, n'}$ reflect different local charge configurations $(n,\ n')$.
Therefore, we can identify their basis-independent eigenvalues (their OSES values) with the charge configurations.
We have shown that the blocks $\C_{n,n'}$ and $\C_{n',n}$ for $n\neq n'$ lead to the same eigenvalues and that therefore the OSES is degenerate.
Finally, we sum up the degenerate OSES values by summing the traces of the degenerate blocks $\C_{n\neq n'}$ and arrive at the configuration coherence~\eqref{eq: configuration entanglement}.
Importantly, we did not impose any further restriction onto the mixed state except for the presence of a fixed charge.
Moreover, the configuration coherence~\eqref{eq: configuration entanglement} is basis independent, which is more apparent in our alternative definitions~\eqref{eq: alternative version of configuration entanglement} and \eqref{eq: configuration entanglement as relative purity} in appendix~\ref{app: configuration entanglement}.

The measure $\E$ encodes the coherence between sectors of different local particle numbers $(n, \, n')$ and has the following properties:
\begin{enumerate}
    \item\label{property: vanishes for separable} It vanishes for separable states with fixed particle number.
    \item\label{property: convex} It is convex, i.e., $\E(\sum_i p_i \rho_i) \leq \sum_i p_i \E (\rho_i)$ for density matrices $\rho_i$ and $\sum_i p_i=1$.
    \item\label{property: unitary invariance} It is invariant under local particle number conserving unitary operations.
    \item\label{property: montonous under unital} It is monotonously decreasing under local purity non-increasing particle number conserving operations.
\end{enumerate}
Proofs for these statements can be found in appendix~\ref{app: configuration entanglement}.
Let us briefly discuss the implications of these properties.
Property~\ref{property: vanishes for separable} means that the configuration coherence is an entanglement witness for systems with a fixed particle number:
if the state of such a system has non-zero configuration coherence, it cannot be separable, i.e., it is entangled.
Due to property~\ref{property: unitary invariance}, the configuration coherence is independent of the local basis choice and cannot be changed by isolated subsystem evolution.
Property~\ref{property: montonous under unital}, in combination with the ability to witness entanglement, makes the configuration coherence an entanglement measure under purity non-increasing particle number conserving maps~\cite{horodecki_et_al_2009}.
The most general purity non-increasing maps are unital maps, i.e., those that preserve the completely mixed state~\cite{gour_et_al_2015, streltsov_et_al_2018}.
Importantly, these include evolutions governed by a Lindblad master equation with Hermitian jump operators~\cite{manzano_2020}, making the configuration coherence a powerful and versatile quantifier for coherence and entanglement.
The definition of the configuration coherence~\eqref{eq: configuration entanglement} together with its properties is the main result of this work.

The remaining blocks of the $\C$-matrix with $n=n'\notin[0,N]$ describe both classical and quantum correlations. This is evident from their pure limit, where they exhibit both non-degenerate and degenerate eigenvalues, of which only the latter contribute to the measure~\eqref{eq: pure state entanglement measure}, see Fig.~\ref{fig: 2 particle example}(c). Note that when the state is mixed, the degeneracy may be lifted via avoided crossings between classical- and quantum-correlation values (see marker (1) in Fig.~\ref{fig: 2 particle example}(c)).
In this situation, it is in general not possible to assign the OSES values of such blocks to classical or quantum correlations.
We, therefore, do not include these values in the configuration coherence~\eqref{eq: configuration entanglement}.
The OSES values of the $n=n'\notin [0, N]$ mix information about classical and quantum correlations within the $n=n'$ charge sector, and the classical correlations would tamper with the desirable properties of the configuration coherence.
Specifically, the configuration coherence would neither be an entanglement witness (classically correlated states would have non-zero configuration coherence) nor an entanglement monotone (by increasing classical correlations we could increase the configuration coherence).
Distilling the amount of quantum correlations contained in such OSES values will be the topic of future work.


\begin{figure}[t!]
    \centering
    \includegraphics[width=\columnwidth]{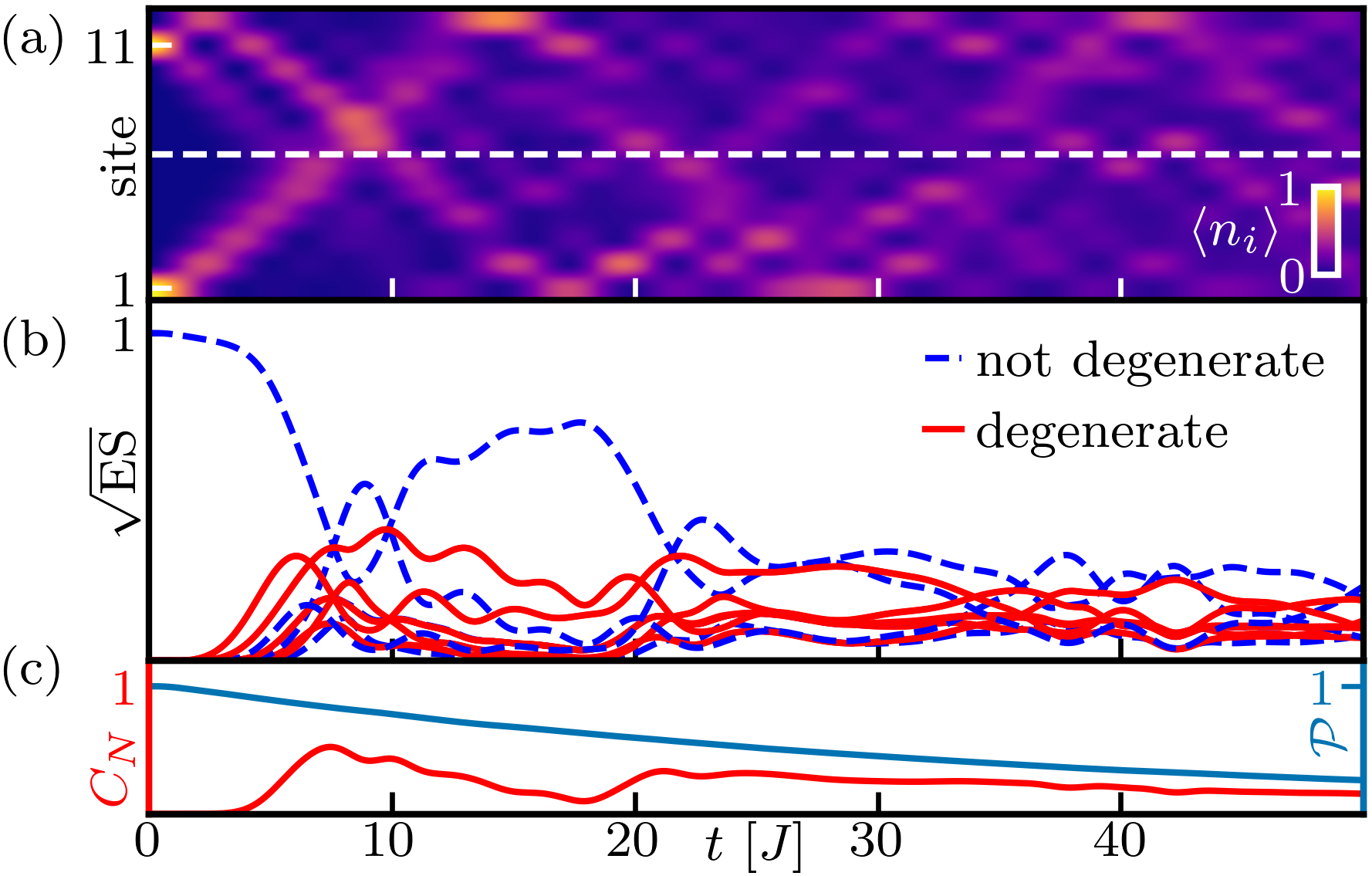}
    \caption{ Two particles hopping on a 12-site chain while subject to dephasing [cf.~Eq.~\eqref{eq: Lindblad}], implemented in MPDO-TEBD with timestep $dt=0.05J$, dephasing strength $\gamma=0.005J$, and maximal bond dimension $\chi=200$.
    (a) Local particle densities $\langle n_i \rangle$ as a function of space along time. Dashed white line denotes bipartition at the middle bond.
    (b) Evolution of the OSES at the bipartition. Blue-dashed lines: non-degenerate OSES values. Red lines: degenerate OSES values encoding the cross-boundary coherence.
    (c) Configuration coherence $\E$ and purity $\P$ calculated as the sum over the degenerate and all OSES values, respectively.}
    \label{fig: algorithm}
\end{figure}

The remaining challenge involves the efficient spectral filtering of a density matrix $\rho$ describing a mixed state of a realistic open quantum system with fixed particle number. For one dimensional systems, we propose to compute the configuration coherence using a matrix product density operator (MPDO) decomposition of $\rho$~\cite{verstraete_et_al_2004} [cf.~Fig.~\ref{fig: 2 particle example}(d)].
The MPDO description gives direct and efficient access to the OSES~\cite{schollwock_2011}. 
In principle, the configuration coherence~\eqref{eq: configuration entanglement} can then be obtained as the sum over the degenerate values. 
However, there might be accidental degeneracies due to level crossings in the spectrum, see~Fig.~\ref{fig: 2 particle example}(b). 
Such degeneracies can be revealed by a small continuous deformation of the MPDO, e.g., using an interpolation of the form
\begin{equation}
\label{eq: mixed state interpolation}
    \rho_p = (1-p)\rho + p {\rm diag}(\rho)\,.
\end{equation}
This interpolation can be applied independent of the representation of the density matrix, and it is straightforwardly implemented within the MPDO formalism:  
The diagonal state ${\rm diag}(\rho)$ is obtained by setting the off-diagonal elements of the local MPDO-matrices to zero and the addition is an efficient operation for MPDOs (the bond dimension of a sum of MPDOs is bounded by the sum of the individual bond dimensions).
It is important to note that the computability of the configuration coherence from an MPDO-representation of a mixed state at no additional computational cost is a major advantage of our approach. 
In comparison, the calculation of the negativity requires additional diagonalization of an exponentially large operator.

We turn now to showcase the configuration coherence in a realistic open system scenario. We consider a system of $N=2$ spinless particles moving on a 1D chain with 12 sites in the presence of dephasing. For this system, we can readily obtain an exact MPDO representation because of the small rank of the $\C$-matrix, see appendix~\ref{app: rank of C-matrix}. For larger systems requiring a truncated MPDO description, our algorithm will yield a tight lower bound on the configuration coherence. The time evolution follows a Lindblad master equation
\begin{equation}
	\label{eq: Lindblad}
	\partial_t \rho =  - i [H, \rho] + \gamma \sum_i  \left(2 n_i\rho n_i -  \left\{n_i, \rho\right\}\right)\,,
\end{equation}
with a hopping Hamiltonian 
\begin{equation}
H = J\sum_i c_i^\dagger c_{i+1} + h.c.\,, 
\end{equation}
and local density operators $n_i = c_i^\dagger c_i$. The parameters $J$ and $\gamma$ are the hopping amplitude and the dephasing coupling rate to local baths, respectively.
As discussed above, by property~\ref{property: montonous under unital} and the unitality of the evolution~\eqref{eq: Lindblad}, it follows that the configuration coherence is a faithful entanglement measure for this example.

We initialise the system in a product state of one particle on site 1 and the other on site 11, and evolve Eq.~\eqref{eq: Lindblad} using time evolving block decimation (TEBD) with the \textsc{Julia} ITensors package~\cite{fishman_et_al_2021}. In
Fig.~\ref{fig: algorithm}(a), we present the resulting density of the two particles.
The OSES associated with a half-chain bipartition is directly obtained from the MPDO representation throughout the time evolution, see Fig.~\ref{fig: algorithm}(b).
As expected for a product state, the OSES at $t~=~0$ consists of a single value only, describing one particle in each of the subsystems to the left and right of the cut. Along the time evolution, the particles delocalize across the cut, leading to cross-boundary coherence and entanglement, evident by the emergence of degenerate OSES values.
We extract the configuration coherence~\eqref{eq: configuration entanglement} as the sum over the degenerate OSES values and the purity as the sum over all OSES values, see  Fig.~\ref{fig: algorithm}(c). 
The dissipative Lindblad terms continuously decrease the purity.

In conclusion, we have analyzed the OSES for general pure states as well as mixed states with a fixed particle number.
For pure states, the sum over degenerate OSES values lends a quantifier of quantum correlations that is closely related to the negativity.
For mixed states with fixed particle number, we defined the sum over degenerate OSES values as the configuration coherence.
The configuration coherence is a basis-independent witness of entanglement between sectors with different local particle numbers.
Moreover, for purity non-increasing particle number conserving maps such as Lindblad-type evolutions with Hermitian jump operators, the configuration coherence is an entanglement measure.
As an example, we have measured the entanglement of spinless particles using a MPDO-TEBD algorithm. 
Thus, we have shown that the configuration coherence can be efficiently computed for 1D systems with a suitable low-rank MPDO representation.
Such systems include infinite size dissipative quantum chains~\cite{gangat_et_al_2017}, open many-body localized systems~\cite{fischer_et_al_2016, van_nieuwenburg_et_al_2017, lenarcic_et_al_2020}, strongly thermalizing systems~\cite{white_et_al_2018}, exciton dynamics~\cite{jaschke_et_al_2018}, the quantum Heisenberg magnet~\cite{dupont_et_al_2021}, and temporal entanglement in many-body Floquet dynamics~\cite{lerose_et_al_2021, sonner_et_al_2021}.
Experimentally, the configuration coherence can be obtained by estimating the purity of the mixed state~\cite{lukin_et_al_2019, brydges_et_al_2019} and subtracting the values encoding classical correlations; the latter are constructed out of local density measurements.
Our results facilitate the study of entanglement and coherence in contemporary noisy intermediate-scale quantum era systems~\cite{preskill_2018, brooks_2019} and motivate further OSES-based measures and complexity estimates.
A natural extension of our work will involve the potential of discarding the fixed charge assumption for the mixed state.
In a first step, one could calculate perturbative corrections to the degenerate OSES values if a second charge sector weakly contributes to the mixed state.
We expect the perturbation to lift the degeneracies.
This will also happen when there is no symmetry in the system.
In this case, the degeneracy structure of the OSES might contain different information.

\emph{Acknowledgments.} We thank E.~van  Nieuwenburg, M.~S.~Ferguson, A.~\v{S}trkalj, M.~H.~Fischer, L.~Stocker, A.~Romito, J.~del~Pino, D.~Sutter, J.~Renes, M.~Michalek, and S.~Ryu for fruitful discussions. We thank Z.~Ma, C.~Han, and E.~Sela for pointing out a wrong statement in an earlier version of this manuscript and for drawing our attention to their related work on number entanglement~\cite{ma_et_al_2022}. We acknowledge financial support by ETH Research Grant ETH-51 20-1 and the Deutsche Forschungsgemeinschaft (DFG) - project number 449653034.

\bibliographystyle{unsrtnat}
\bibliography{entanglement}

\begin{thebibliography}{88}
\providecommand{\natexlab}[1]{#1}
\providecommand{\url}[1]{\texttt{#1}}
\expandafter\ifx\csname urlstyle\endcsname\relax
  \providecommand{\doi}[1]{doi: #1}\else
  \providecommand{\doi}{doi: \begingroup \urlstyle{rm}\Url}\fi

\bibitem[Nielsen and Chuang(2010)]{nielsen_chuang_2010}
Michael~A. Nielsen and Isaac~L. Chuang.
\newblock \emph{Quantum Computation and Quantum Information: 10th Anniversary
  Edition}.
\newblock Cambridge University Press, 2010.
\newblock \doi{10.1017/CBO9780511976667}.

\bibitem[Boixo et~al.(2018)Boixo, Isakov, Smelyanskiy, Babbush, Ding, Jiang,
  Bremner, Martinis, and Neven]{boixo_et_al_2018}
Sergio Boixo, Sergei~V. Isakov, Vadim~N. Smelyanskiy, Ryan Babbush, Nan Ding,
  Zhang Jiang, Michael~J. Bremner, John~M. Martinis, and Hartmut Neven.
\newblock Characterizing quantum supremacy in near-term devices.
\newblock \emph{Nature Physics}, 14\penalty0 (66):\penalty0 595–600, Jun
  2018.
\newblock ISSN 1745-2481.
\newblock \doi{10.1038/s41567-018-0124-x}.

\bibitem[Neill et~al.(2018)Neill, Roushan, Kechedzhi, Boixo, Isakov,
  Smelyanskiy, Megrant, Chiaro, Dunsworth, Arya, Barends, Burkett, Chen, Chen,
  Fowler, Foxen, Giustina, Graff, Jeffrey, Huang, Kelly, Klimov, Lucero, Mutus,
  Neeley, Quintana, Sank, Vainsencher, Wenner, White, Neven, and
  Martinis]{neill_et_al_2018}
C.~Neill, P.~Roushan, K.~Kechedzhi, S.~Boixo, S.~V. Isakov, V.~Smelyanskiy,
  A.~Megrant, B.~Chiaro, A.~Dunsworth, K.~Arya, R.~Barends, B.~Burkett,
  Y.~Chen, Z.~Chen, A.~Fowler, B.~Foxen, M.~Giustina, R.~Graff, E.~Jeffrey,
  T.~Huang, J.~Kelly, P.~Klimov, E.~Lucero, J.~Mutus, M.~Neeley, C.~Quintana,
  D.~Sank, A.~Vainsencher, J.~Wenner, T.~C. White, H.~Neven, and J.~M.
  Martinis.
\newblock A blueprint for demonstrating quantum supremacy with superconducting
  qubits.
\newblock \emph{Science}, 360\penalty0 (6385):\penalty0 195–199, Apr 2018.
\newblock \doi{10.1126/science.aao4309}.

\bibitem[Arute et~al.(2019)Arute, Arya, Babbush, Bacon, Bardin, Barends,
  Biswas, Boixo, Brandao, Buell, Burkett, Chen, Chen, Chiaro, Collins,
  Courtney, Dunsworth, Farhi, Foxen, Fowler, Gidney, Giustina, Graff, Guerin,
  Habegger, Harrigan, Hartmann, Ho, Hoffmann, Huang, Humble, Isakov, Jeffrey,
  Jiang, Kafri, Kechedzhi, Kelly, Klimov, Knysh, Korotkov, Kostritsa, Landhuis,
  Lindmark, Lucero, Lyakh, Mandrà, McClean, McEwen, Megrant, Mi, Michielsen,
  Mohseni, Mutus, Naaman, Neeley, Neill, Niu, Ostby, Petukhov, Platt, Quintana,
  Rieffel, Roushan, Rubin, Sank, Satzinger, Smelyanskiy, Sung, Trevithick,
  Vainsencher, Villalonga, White, Yao, Yeh, Zalcman, Neven, and
  Martinis]{arute_et_al_2019}
Frank Arute, Kunal Arya, Ryan Babbush, Dave Bacon, Joseph~C. Bardin, Rami
  Barends, Rupak Biswas, Sergio Boixo, Fernando G. S.~L. Brandao, David~A.
  Buell, Brian Burkett, Yu~Chen, Zijun Chen, Ben Chiaro, Roberto Collins,
  William Courtney, Andrew Dunsworth, Edward Farhi, Brooks Foxen, Austin
  Fowler, Craig Gidney, Marissa Giustina, Rob Graff, Keith Guerin, Steve
  Habegger, Matthew~P. Harrigan, Michael~J. Hartmann, Alan Ho, Markus Hoffmann,
  Trent Huang, Travis~S. Humble, Sergei~V. Isakov, Evan Jeffrey, Zhang Jiang,
  Dvir Kafri, Kostyantyn Kechedzhi, Julian Kelly, Paul~V. Klimov, Sergey Knysh,
  Alexander Korotkov, Fedor Kostritsa, David Landhuis, Mike Lindmark, Erik
  Lucero, Dmitry Lyakh, Salvatore Mandrà, Jarrod~R. McClean, Matthew McEwen,
  Anthony Megrant, Xiao Mi, Kristel Michielsen, Masoud Mohseni, Josh Mutus,
  Ofer Naaman, Matthew Neeley, Charles Neill, Murphy~Yuezhen Niu, Eric Ostby,
  Andre Petukhov, John~C. Platt, Chris Quintana, Eleanor~G. Rieffel, Pedram
  Roushan, Nicholas~C. Rubin, Daniel Sank, Kevin~J. Satzinger, Vadim
  Smelyanskiy, Kevin~J. Sung, Matthew~D. Trevithick, Amit Vainsencher, Benjamin
  Villalonga, Theodore White, Z.~Jamie Yao, Ping Yeh, Adam Zalcman, Hartmut
  Neven, and John~M. Martinis.
\newblock Quantum supremacy using a programmable superconducting processor.
\newblock \emph{Nature}, 574\penalty0 (77797779):\penalty0 505–510, Oct 2019.
\newblock ISSN 1476-4687.
\newblock \doi{10.1038/s41586-019-1666-5}.

\bibitem[Bennett et~al.(1996)Bennett, DiVincenzo, Smolin, and
  Wootters]{bennett_et_al_1996}
Charles~H. Bennett, David~P. DiVincenzo, John~A. Smolin, and William~K.
  Wootters.
\newblock Mixed-state entanglement and quantum error correction.
\newblock \emph{Phys. Rev. A}, 54:\penalty0 3824--3851, Nov 1996.
\newblock \doi{10.1103/PhysRevA.54.3824}.

\bibitem[Cory et~al.(1998)Cory, Price, Maas, Knill, Laflamme, Zurek, Havel, and
  Somaroo]{cory_et_al_1998}
D.~G. Cory, M.~D. Price, W.~Maas, E.~Knill, R.~Laflamme, W.~H. Zurek, T.~F.
  Havel, and S.~S. Somaroo.
\newblock Experimental quantum error correction.
\newblock \emph{Phys. Rev. Lett.}, 81:\penalty0 2152--2155, Sep 1998.
\newblock \doi{10.1103/PhysRevLett.81.2152}.

\bibitem[Schindler et~al.(2011)Schindler, Barreiro, Monz, Nebendahl, Nigg,
  Chwalla, Hennrich, and Blatt]{schindler_et_al_2011}
Philipp Schindler, Julio~T. Barreiro, Thomas Monz, Volckmar Nebendahl, Daniel
  Nigg, Michael Chwalla, Markus Hennrich, and Rainer Blatt.
\newblock Experimental repetitive quantum error correction.
\newblock \emph{Science}, 332\penalty0 (6033):\penalty0 1059–1061, May 2011.
\newblock \doi{10.1126/science.1203329}.

\bibitem[Andersen et~al.(2020)Andersen, Remm, Lazar, Krinner, Lacroix, Norris,
  Gabureac, Eichler, and Wallraff]{andersen_et_al_2020}
Christian~Kraglund Andersen, Ants Remm, Stefania Lazar, Sebastian Krinner,
  Nathan Lacroix, Graham~J. Norris, Mihai Gabureac, Christopher Eichler, and
  Andreas Wallraff.
\newblock Repeated quantum error detection in a surface code.
\newblock \emph{Nature Physics}, 16\penalty0 (8):\penalty0 875–880, Aug 2020.
\newblock ISSN 1745-2481.
\newblock \doi{10.1038/s41567-020-0920-y}.

\bibitem[Krinner et~al.(2022)Krinner, Lacroix, Remm, Di~Paolo, Genois, Leroux,
  Hellings, Lazar, Swiadek, Herrmann, Norris, Andersen, Müller, Blais,
  Eichler, and Wallraff]{krinner_et_al_2021}
Sebastian Krinner, Nathan Lacroix, Ants Remm, Agustin Di~Paolo, Elie Genois,
  Catherine Leroux, Christoph Hellings, Stefania Lazar, Francois Swiadek,
  Johannes Herrmann, Graham~J. Norris, Christian~Kraglund Andersen, Markus
  Müller, Alexandre Blais, Christopher Eichler, and Andreas Wallraff.
\newblock Realizing repeated quantum error correction in a distance-three
  surface code.
\newblock \emph{Nature}, 605\penalty0 (7911):\penalty0 669–674, May 2022.
\newblock ISSN 1476-4687.
\newblock \doi{10.1038/s41586-022-04566-8}.

\bibitem[Pezz{\'e} and Smerzi(2009)]{pezze_smerzi_2009}
Luca Pezz{\'e} and Augusto Smerzi.
\newblock Entanglement, nonlinear dynamics, and the heisenberg limit.
\newblock \emph{Phys. Rev. Lett.}, 102:\penalty0 100401, Mar 2009.
\newblock \doi{10.1103/PhysRevLett.102.100401}.

\bibitem[Demkowicz-Dobrza{\'o}ski et~al.(2012)Demkowicz-Dobrza{\'o}ski,
  Ko{\l}ody{\'n}ski, and Gu{\c o}{\v o}]{demkowicz_et_al_2012}
Rafa{\l} Demkowicz-Dobrza{\'o}ski, Jan Ko{\l}ody{\'n}ski, and M{\v a}d{\v a}lin
  Gu{\c o}{\v o}.
\newblock The elusive heisenberg limit in quantum-enhanced metrology.
\newblock \emph{Nature Communications}, 3\penalty0 (11):\penalty0 1063, Sep
  2012.
\newblock ISSN 2041-1723.
\newblock \doi{10.1038/ncomms2067}.

\bibitem[Zhou et~al.(2018)Zhou, Zhang, Preskill, and Jiang]{zhou_et_al_2018}
Sisi Zhou, Mengzhen Zhang, John Preskill, and Liang Jiang.
\newblock Achieving the heisenberg limit in quantum metrology using quantum
  error correction.
\newblock \emph{Nature Communications}, 9\penalty0 (11):\penalty0 78, Jan 2018.
\newblock ISSN 2041-1723.
\newblock \doi{10.1038/s41467-017-02510-3}.

\bibitem[Long et~al.(2007)Long, Deng, Wang, Li, Wen, and Wang]{long_et_al_2007}
Gui-lu Long, Fu-guo Deng, Chuan Wang, Xi-han Li, Kai Wen, and Wan-ying Wang.
\newblock Quantum secure direct communication and deterministic secure quantum
  communication.
\newblock \emph{Frontiers of Physics in China}, 2\penalty0 (3):\penalty0
  251–272, Jul 2007.
\newblock ISSN 1673-3606.
\newblock \doi{10.1007/s11467-007-0050-3}.

\bibitem[Hu et~al.(2016)Hu, Yu, Jing, Xiao, Jia, Qin, and Long]{hu_et_al_2016}
Jian-Yong Hu, Bo~Yu, Ming-Yong Jing, Lian-Tuan Xiao, Suo-Tang Jia, Guo-Qing
  Qin, and Gui-Lu Long.
\newblock Experimental quantum secure direct communication with single photons.
\newblock \emph{Light: Science \& Applications}, 5\penalty0 (99):\penalty0
  e16144–e16144, Sep 2016.
\newblock ISSN 2047-7538.
\newblock \doi{10.1038/lsa.2016.144}.

\bibitem[Zhang et~al.(2017)Zhang, Ding, Sheng, Zhou, Shi, and
  Guo]{zhang_et_al_2017}
Wei Zhang, Dong-Sheng Ding, Yu-Bo Sheng, Lan Zhou, Bao-Sen Shi, and Guang-Can
  Guo.
\newblock Quantum secure direct communication with quantum memory.
\newblock \emph{Phys. Rev. Lett.}, 118:\penalty0 220501, May 2017.
\newblock \doi{10.1103/PhysRevLett.118.220501}.

\bibitem[Bouwmeester et~al.(1997)Bouwmeester, Pan, Mattle, Eibl, Weinfurter,
  and Zeilinger]{bouwmeester_et_al_1997}
Dik Bouwmeester, Jian-Wei Pan, Klaus Mattle, Manfred Eibl, Harald Weinfurter,
  and Anton Zeilinger.
\newblock Experimental quantum teleportation.
\newblock \emph{Nature}, 390\penalty0 (66606660):\penalty0 575–579, Dec 1997.
\newblock ISSN 1476-4687.
\newblock \doi{10.1038/37539}.

\bibitem[Furusawa et~al.(1998)Furusawa, Sørensen, Braunstein, Fuchs, Kimble,
  and Polzik]{furusawa_et_al_1998}
A.~Furusawa, J.~L. Sørensen, S.~L. Braunstein, C.~A. Fuchs, H.~J. Kimble, and
  E.~S. Polzik.
\newblock Unconditional quantum teleportation.
\newblock \emph{Science}, 282\penalty0 (5389):\penalty0 706–709, Oct 1998.
\newblock \doi{10.1126/science.282.5389.706}.

\bibitem[Nielsen et~al.(1998)Nielsen, Knill, and Laflamme]{nielsen_et_al_1998}
M.~A. Nielsen, E.~Knill, and R.~Laflamme.
\newblock Complete quantum teleportation using nuclear magnetic resonance.
\newblock \emph{Nature}, 396\penalty0 (67066706):\penalty0 52–55, Nov 1998.
\newblock ISSN 1476-4687.
\newblock \doi{10.1038/23891}.

\bibitem[Riebe et~al.(2004)Riebe, Häffner, Roos, Hänsel, Benhelm, Lancaster,
  Körber, Becher, Schmidt-Kaler, James, and Blatt]{riebe_et_al_2004}
M.~Riebe, H.~Häffner, C.~F. Roos, W.~Hänsel, J.~Benhelm, G.~P.~T. Lancaster,
  T.~W. Körber, C.~Becher, F.~Schmidt-Kaler, D.~F.~V. James, and R.~Blatt.
\newblock Deterministic quantum teleportation with atoms.
\newblock \emph{Nature}, 429\penalty0 (69936993):\penalty0 734–737, Jun 2004.
\newblock ISSN 1476-4687.
\newblock \doi{10.1038/nature02570}.

\bibitem[Nozi{\`e}res and Blandin(1980)]{nozieres_blandin_1980}
Ph. Nozi{\`e}res and Annie Blandin.
\newblock {Kondo effect in real metals}.
\newblock \emph{{Journal de Physique}}, 41\penalty0 (3):\penalty0 193--211,
  1980.
\newblock \doi{10.1051/jphys:01980004103019300}.

\bibitem[Kondo(2012)]{kondo_2012}
Jun Kondo.
\newblock \emph{The Physics of Dilute Magnetic Alloys}.
\newblock Cambridge University Press, 2012.
\newblock \doi{10.1017/CBO9781139162173}.

\bibitem[Basko et~al.(2006)Basko, Aleiner, and Altshuler]{basko_et_al_2006}
D.~M. Basko, I.~L. Aleiner, and B.~L. Altshuler.
\newblock Metal–insulator transition in a weakly interacting many-electron
  system with localized single-particle states.
\newblock \emph{Annals of Physics}, 321\penalty0 (5):\penalty0 1126–1205,
  2006.
\newblock ISSN 0003-4916.
\newblock \doi{10.1016/j.aop.2005.11.014}.

\bibitem[Nandkishore and Huse(2015)]{nandkishore_huse_2015}
Rahul Nandkishore and David~A. Huse.
\newblock Many-body localization and thermalization in quantum statistical
  mechanics.
\newblock \emph{Annual Review of Condensed Matter Physics}, 6\penalty0
  (1):\penalty0 15--38, 2015.
\newblock \doi{10.1146/annurev-conmatphys-031214-014726}.

\bibitem[Stormer et~al.(1999)Stormer, Tsui, and Gossard]{stormer_et_al_1999}
Horst~L. Stormer, Daniel~C. Tsui, and Arthur~C. Gossard.
\newblock The fractional quantum hall effect.
\newblock \emph{Rev. Mod. Phys.}, 71:\penalty0 S298--S305, Mar 1999.
\newblock \doi{10.1103/RevModPhys.71.S298}.

\bibitem[Avella and Mancini(2012)]{avella_mancini_2012}
Adolfo Avella and Ferdinando Mancini.
\newblock \emph{Strongly Correlated Systems: Theoretical Methods}.
\newblock Springer, Berlin Heidelberg, 01 2012.
\newblock ISBN 978-3-642-21830-9.
\newblock \doi{10.1007/978-3-642-21831-6}.

\bibitem[Bruus and Flensberg(2004)]{bruus_flensberg_2004}
Henrik Bruus and Karsten Flensberg.
\newblock \emph{Many-body quantum theory in condensed matter physics: an
  introduction}.
\newblock OUP Oxford, 2004.
\newblock ISBN 978-0-19-856633-5.

\bibitem[Carusotto and Ciuti(2013)]{carusotto_ciuti_2013}
Iacopo Carusotto and Cristiano Ciuti.
\newblock Quantum fluids of light.
\newblock \emph{Rev. Mod. Phys.}, 85:\penalty0 299--366, Feb 2013.
\newblock \doi{10.1103/RevModPhys.85.299}.

\bibitem[Bloch et~al.(2008)Bloch, Dalibard, and Zwerger]{bloch_et_al_2008}
Immanuel Bloch, Jean Dalibard, and Wilhelm Zwerger.
\newblock Many-body physics with ultracold gases.
\newblock \emph{Rev. Mod. Phys.}, 80:\penalty0 885--964, Jul 2008.
\newblock \doi{10.1103/RevModPhys.80.885}.

\bibitem[Campagnano et~al.(2012)Campagnano, Zilberberg, Gornyi, Feldman,
  Potter, and Gefen]{campagnano_et_al_2012}
Gabriele Campagnano, Oded Zilberberg, Igor~V. Gornyi, Dmitri~E. Feldman,
  Andrew~C. Potter, and Yuval Gefen.
\newblock Hanbury brown--twiss interference of anyons.
\newblock \emph{Phys. Rev. Lett.}, 109:\penalty0 106802, Sep 2012.
\newblock \doi{10.1103/PhysRevLett.109.106802}.

\bibitem[Shapourian et~al.(2017)Shapourian, Shiozaki, and
  Ryu]{shapourian_et_al_2017}
Hassan Shapourian, Ken Shiozaki, and Shinsei Ryu.
\newblock Partial time-reversal transformation and entanglement negativity in
  fermionic systems.
\newblock \emph{Phys. Rev. B}, 95:\penalty0 165101, Apr 2017.
\newblock \doi{10.1103/PhysRevB.95.165101}.

\bibitem[Wolf et~al.(2019)Wolf, Lado, Blatter, and Zilberberg]{wolf_et_al_2019}
T.~M.~R. Wolf, J.~L. Lado, G.~Blatter, and O.~Zilberberg.
\newblock Electrically tunable flat bands and magnetism in twisted bilayer
  graphene.
\newblock \emph{Phys. Rev. Lett.}, 123:\penalty0 096802, Aug 2019.
\newblock \doi{10.1103/PhysRevLett.123.096802}.

\bibitem[Wolf et~al.(2021)Wolf, Zilberberg, Blatter, and Lado]{wolf_et_al_2021}
Tobias M.~R. Wolf, Oded Zilberberg, Gianni Blatter, and Jose~L. Lado.
\newblock Spontaneous valley spirals in magnetically encapsulated twisted
  bilayer graphene.
\newblock \emph{Phys. Rev. Lett.}, 126:\penalty0 056803, Feb 2021.
\newblock \doi{10.1103/PhysRevLett.126.056803}.

\bibitem[Lado and Zilberberg(2019)]{lado_zilberberg_2019}
J.~L. Lado and Oded Zilberberg.
\newblock Topological spin excitations in harper-heisenberg spin chains.
\newblock \emph{Phys. Rev. Research}, 1:\penalty0 033009, Oct 2019.
\newblock \doi{10.1103/PhysRevResearch.1.033009}.

\bibitem[\ifmmode~\check{S}\else \v{S}\fi{}trkalj
  et~al.(2021)\ifmmode~\check{S}\else \v{S}\fi{}trkalj, Doggen, Gornyi, and
  Zilberberg]{strkalj_et_al_2021}
Antonio \ifmmode~\check{S}\else \v{S}\fi{}trkalj, Elmer V.~H. Doggen, Igor~V.
  Gornyi, and Oded Zilberberg.
\newblock Many-body localization in the interpolating aubry-andr\'e-fibonacci
  model.
\newblock \emph{Phys. Rev. Research}, 3:\penalty0 033257, Sep 2021.
\newblock \doi{10.1103/PhysRevResearch.3.033257}.

\bibitem[Khedri et~al.(2021)Khedri, \ifmmode~\check{S}\else \v{S}\fi{}trkalj,
  Chiocchetta, and Zilberberg]{khedri_et_al_2021}
Andisheh Khedri, Antonio \ifmmode~\check{S}\else \v{S}\fi{}trkalj, Alessio
  Chiocchetta, and Oded Zilberberg.
\newblock Luttinger liquid coupled to ohmic-class environments.
\newblock \emph{Phys. Rev. Research}, 3:\penalty0 L032013, Jul 2021.
\newblock \doi{10.1103/PhysRevResearch.3.L032013}.

\bibitem[Ferguson et~al.(2020)Ferguson, Camenzind, Müller, Biesinger,
  Scheller, Braunecker, Zumbühl, and Zilberberg]{ferguson_et_al_2020}
Michael~S. Ferguson, Leon~C. Camenzind, Clemens Müller, Daniel E.~F.
  Biesinger, Christian~P. Scheller, Bernd Braunecker, Dominik~M. Zumbühl, and
  Oded Zilberberg.
\newblock Quantum measurement induces a many-body transition.
\newblock \emph{arXiv:2010.04635 [cond-mat]}, Oct 2020.
\newblock \doi{10.48550/ARXIV.2010.04635}.

\bibitem[Ferguson et~al.(2021)Ferguson, Zilberberg, and
  Blatter]{ferguson_et_al_2021}
Michael~Sven Ferguson, Oded Zilberberg, and Gianni Blatter.
\newblock Open quantum systems beyond fermi's golden rule: Diagrammatic
  expansion of the steady-state time-convolutionless master equations.
\newblock \emph{Phys. Rev. Research}, 3:\penalty0 023127, May 2021.
\newblock \doi{10.1103/PhysRevResearch.3.023127}.

\bibitem[Li et~al.(2018)Li, Chen, and Fisher]{li_et_al_2018}
Yaodong Li, Xiao Chen, and Matthew P.~A. Fisher.
\newblock Quantum zeno effect and the many-body entanglement transition.
\newblock \emph{Phys. Rev. B}, 98:\penalty0 205136, Nov 2018.
\newblock \doi{10.1103/PhysRevB.98.205136}.

\bibitem[Szyniszewski et~al.(2019)Szyniszewski, Romito, and
  Schomerus]{szyniszewski_et_al_2019}
M.~Szyniszewski, A.~Romito, and H.~Schomerus.
\newblock Entanglement transition from variable-strength weak measurements.
\newblock \emph{Phys. Rev. B}, 100:\penalty0 064204, Aug 2019.
\newblock \doi{10.1103/PhysRevB.100.064204}.

\bibitem[Liu et~al.(2022)Liu, Sohal, Kudler-Flam, and Ryu]{liu_et_al_2022}
Yuhan Liu, Ramanjit Sohal, Jonah Kudler-Flam, and Shinsei Ryu.
\newblock Multipartitioning topological phases by vertex states and quantum
  entanglement.
\newblock \emph{Phys. Rev. B}, 105:\penalty0 115107, Mar 2022.
\newblock \doi{10.1103/PhysRevB.105.115107}.

\bibitem[Sarovar et~al.(2010)Sarovar, Ishizaki, Fleming, and
  Whaley]{sarovar_et_al_2010}
Mohan Sarovar, Akihito Ishizaki, Graham~R. Fleming, and K.~Birgitta Whaley.
\newblock Quantum entanglement in photosynthetic light-harvesting complexes.
\newblock \emph{Nature Physics}, 6\penalty0 (66):\penalty0 462–467, Jun 2010.
\newblock ISSN 1745-2481.
\newblock \doi{10.1038/nphys1652}.

\bibitem[Caruso et~al.(2010)Caruso, Chin, Datta, Huelga, and
  Plenio]{caruso_et_al_2010}
Filippo Caruso, Alex~W. Chin, Animesh Datta, Susana~F. Huelga, and Martin~B.
  Plenio.
\newblock Entanglement and entangling power of the dynamics in light-harvesting
  complexes.
\newblock \emph{Phys. Rev. A}, 81:\penalty0 062346, Jun 2010.
\newblock \doi{10.1103/PhysRevA.81.062346}.

\bibitem[Ishizaki and Fleming(2010)]{ishizaki_fleming_2010}
Akihito Ishizaki and Graham~R Fleming.
\newblock Quantum superpositions in photosynthetic light harvesting:
  delocalization and entanglement.
\newblock \emph{New Journal of Physics}, 12\penalty0 (5):\penalty0 055004, may
  2010.
\newblock \doi{10.1088/1367-2630/12/5/055004}.

\bibitem[van Nieuwenburg and Zilberberg(2018)]{nieuwenburg_zilberberg_2018}
Evert van Nieuwenburg and Oded Zilberberg.
\newblock Entanglement spectrum of mixed states.
\newblock \emph{Phys. Rev. A}, 98:\penalty0 012327, Jul 2018.
\newblock \doi{10.1103/PhysRevA.98.012327}.

\bibitem[Stocker et~al.(2022)Stocker, Sack, Ferguson, and
  Zilberberg]{stocker_et_al_2022}
Lidia Stocker, Stefan~H. Sack, Michael~S. Ferguson, and Oded Zilberberg.
\newblock Entanglement-based observables for quantum impurities.
\newblock \emph{Phys. Rev. Res.}, 4:\penalty0 043177, Dec 2022.
\newblock \doi{10.1103/PhysRevResearch.4.043177}.

\bibitem[Perez-Garcia et~al.(2007)Perez-Garcia, Verstraete, Wolf, and
  Cirac]{perez_et_al_2007}
D.~Perez-Garcia, F.~Verstraete, M.~M. Wolf, and J.~I. Cirac.
\newblock Matrix product state representations.
\newblock \emph{arXiv:quant-ph/0608197}, May 2007.
\newblock \doi{10.48550/ARXIV.QUANT-PH/0608197}.

\bibitem[Schollw\"ock(2005)]{schollwock_2005}
U.~Schollw\"ock.
\newblock The density-matrix renormalization group.
\newblock \emph{Rev. Mod. Phys.}, 77:\penalty0 259--315, Apr 2005.
\newblock \doi{10.1103/RevModPhys.77.259}.

\bibitem[Schollw\"ock(2011)]{schollwock_2011}
Ulrich Schollw\"ock.
\newblock The density-matrix renormalization group in the age of matrix product
  states.
\newblock \emph{Annals of Physics}, 326\penalty0 (1):\penalty0 96–192, Jan
  2011.
\newblock ISSN 00034916.
\newblock \doi{10.1016/j.aop.2010.09.012}.

\bibitem[Islam et~al.(2015)Islam, Ma, Preiss, Eric~Tai, Lukin, Rispoli, and
  Greiner]{islam_et_al_2015}
Rajibul Islam, Ruichao Ma, Philipp~M. Preiss, M.~Eric~Tai, Alexander Lukin,
  Matthew Rispoli, and Markus Greiner.
\newblock Measuring entanglement entropy in a quantum many-body system.
\newblock \emph{Nature}, 528\penalty0 (75807580):\penalty0 77–83, Dec 2015.
\newblock ISSN 1476-4687.
\newblock \doi{10.1038/nature15750}.

\bibitem[Gurvits(2003)]{gurvits_2003}
Leonid Gurvits.
\newblock Classical deterministic complexity of edmonds’ problem and quantum
  entanglement.
\newblock \emph{arXiv:quant-ph/0303055}, Mar 2003.
\newblock \doi{10.48550/arXiv.quant-ph/0303055}.

\bibitem[Gharibian(2009)]{gharibian_2008}
Sevag Gharibian.
\newblock Strong np-hardness of the quantum separability problem.
\newblock \emph{arXiv:0810.4507 [quant-ph]}, Dec 2009.
\newblock \doi{10.48550/ARXIV.0810.4507}.

\bibitem[Vidal and Werner(2002)]{vidal_werner_2002}
G.~Vidal and R.~F. Werner.
\newblock Computable measure of entanglement.
\newblock \emph{Phys. Rev. A}, 65:\penalty0 032314, Feb 2002.
\newblock \doi{10.1103/PhysRevA.65.032314}.

\bibitem[Calabrese et~al.(2012)Calabrese, Cardy, and
  Tonni]{calabrese_et_al_2012}
Pasquale Calabrese, John Cardy, and Erik Tonni.
\newblock Entanglement negativity in quantum field theory.
\newblock \emph{Phys. Rev. Lett.}, 109:\penalty0 130502, Sep 2012.
\newblock \doi{10.1103/PhysRevLett.109.130502}.

\bibitem[Calabrese et~al.(2013)Calabrese, Cardy, and
  Tonni]{calabrese_et_al_2013}
Pasquale Calabrese, John Cardy, and Erik Tonni.
\newblock Entanglement negativity in extended systems: a field theoretical
  approach.
\newblock \emph{Journal of Statistical Mechanics: Theory and Experiment},
  2013\penalty0 (02):\penalty0 P02008, Feb 2013.
\newblock ISSN 1742-5468.
\newblock \doi{10.1088/1742-5468/2013/02/P02008}.

\bibitem[Wybo et~al.(2020)Wybo, Knap, and Pollmann]{wybo_et_al_2020}
Elisabeth Wybo, Michael Knap, and Frank Pollmann.
\newblock Entanglement dynamics of a many-body localized system coupled to a
  bath.
\newblock \emph{Phys. Rev. B}, 102:\penalty0 064304, Aug 2020.
\newblock \doi{10.1103/PhysRevB.102.064304}.

\bibitem[Sang et~al.(2021)Sang, Li, Zhou, Chen, Hsieh, and
  Fisher]{sang_et_al_2021}
Shengqi Sang, Yaodong Li, Tianci Zhou, Xiao Chen, Timothy~H. Hsieh, and Matthew
  P.~A. Fisher.
\newblock Entanglement negativity at measurement-induced criticality.
\newblock \emph{PRX Quantum}, 2:\penalty0 030313, Jul 2021.
\newblock \doi{10.1103/PRXQuantum.2.030313}.

\bibitem[Christandl and Winter(2004)]{christandl_winter_2004}
Matthias Christandl and Andreas Winter.
\newblock “squashed entanglement”: An additive entanglement measure.
\newblock \emph{Journal of Mathematical Physics}, 45\penalty0 (3):\penalty0
  829--840, 2004.
\newblock \doi{10.1063/1.1643788}.

\bibitem[Dutta and Faulkner(2021)]{dutta_faulkner_2021}
Souvik Dutta and Thomas Faulkner.
\newblock A canonical purification for the entanglement wedge cross-section.
\newblock \emph{Journal of High Energy Physics}, 2021\penalty0 (3):\penalty0
  178, Mar 2021.
\newblock ISSN 1029-8479.
\newblock \doi{10.1007/JHEP03(2021)178}.

\bibitem[Ma et~al.(2022)Ma, Han, Meir, and Sela]{ma_et_al_2022}
Zhanyu Ma, Cheolhee Han, Yigal Meir, and Eran Sela.
\newblock Symmetric inseparability and number entanglement in charge-conserving
  mixed states.
\newblock \emph{Phys. Rev. A}, 105:\penalty0 042416, Apr 2022.
\newblock \doi{10.1103/PhysRevA.105.042416}.

\bibitem[Zanardi(2001)]{zanardi_2001}
Paolo Zanardi.
\newblock Entanglement of quantum evolutions.
\newblock \emph{Phys. Rev. A}, 63:\penalty0 040304(R), Mar 2001.
\newblock \doi{10.1103/PhysRevA.63.040304}.

\bibitem[Prosen and Pi\ifmmode~\check{z}\else
  \v{z}\fi{}orn(2007)]{prosen_pizorn_2007}
Toma\ifmmode \check{z}\else~\v{z}\fi{} Prosen and Iztok
  Pi\ifmmode~\check{z}\else \v{z}\fi{}orn.
\newblock Operator space entanglement entropy in a transverse ising chain.
\newblock \emph{Phys. Rev. A}, 76:\penalty0 032316, Sep 2007.
\newblock \doi{10.1103/PhysRevA.76.032316}.

\bibitem[Pi\ifmmode~\check{z}\else \v{z}\fi{}orn and
  Prosen(2009)]{pizorn_prosen_2009}
Iztok Pi\ifmmode~\check{z}\else \v{z}\fi{}orn and Toma\ifmmode
  \check{z}\else~\v{z}\fi{} Prosen.
\newblock Operator space entanglement entropy in $xy$ spin chains.
\newblock \emph{Phys. Rev. B}, 79:\penalty0 184416, May 2009.
\newblock \doi{10.1103/PhysRevB.79.184416}.

\bibitem[Li and Haldane(2008)]{li_haldane_2008}
Hui Li and F.~D.~M. Haldane.
\newblock Entanglement spectrum as a generalization of entanglement entropy:
  Identification of topological order in non-abelian fractional quantum hall
  effect states.
\newblock \emph{Phys. Rev. Lett.}, 101:\penalty0 010504, Jul 2008.
\newblock \doi{10.1103/PhysRevLett.101.010504}.

\bibitem[Dubail(2017)]{dubail_2017}
J~Dubail.
\newblock Entanglement scaling of operators: a conformal field theory approach,
  with a glimpse of simulability of long-time dynamics in
  1{\hspace{0.167em}}{\hspace{0.167em}}+{\hspace{0.167em}}{\hspace{0.167em}}1d.
\newblock \emph{Journal of Physics A: Mathematical and Theoretical},
  50\penalty0 (23):\penalty0 234001, may 2017.
\newblock \doi{10.1088/1751-8121/aa6f38}.

\bibitem[van Nieuwenburg and Huber(2014)]{nieuwenburg_huber_2014}
Evert P.~L. van Nieuwenburg and Sebastian~D. Huber.
\newblock Classification of mixed-state topology in one dimension.
\newblock \emph{Phys. Rev. B}, 90:\penalty0 075141, Aug 2014.
\newblock \doi{10.1103/PhysRevB.90.075141}.

\bibitem[Cornfeld et~al.(2018)Cornfeld, Goldstein, and
  Sela]{cornfeld_et_al_2018}
Eyal Cornfeld, Moshe Goldstein, and Eran Sela.
\newblock Imbalance entanglement: Symmetry decomposition of negativity.
\newblock \emph{Phys. Rev. A}, 98:\penalty0 032302, Sep 2018.
\newblock \doi{10.1103/PhysRevA.98.032302}.

\bibitem[Macieszczak et~al.(2019)Macieszczak, Levi, Macr\`{\i}, Lesanovsky, and
  Garrahan]{macieszczak_et_al_2019}
Katarzyna Macieszczak, Emanuele Levi, Tommaso Macr\`{\i}, Igor Lesanovsky, and
  Juan~P. Garrahan.
\newblock Coherence, entanglement, and quantumness in closed and open systems
  with conserved charge, with an application to many-body localization.
\newblock \emph{Phys. Rev. A}, 99:\penalty0 052354, May 2019.
\newblock \doi{10.1103/PhysRevA.99.052354}.

\bibitem[Horodecki et~al.(2009)Horodecki, Horodecki, Horodecki, and
  Horodecki]{horodecki_et_al_2009}
Ryszard Horodecki, Pawe\l{} Horodecki, Micha\l{} Horodecki, and Karol
  Horodecki.
\newblock Quantum entanglement.
\newblock \emph{Rev. Mod. Phys.}, 81:\penalty0 865--942, Jun 2009.
\newblock \doi{10.1103/RevModPhys.81.865}.

\bibitem[Gour et~al.(2015)Gour, Müller, Narasimhachar, Spekkens, and {Yunger
  Halpern}]{gour_et_al_2015}
Gilad Gour, Markus~P. Müller, Varun Narasimhachar, Robert~W. Spekkens, and
  Nicole {Yunger Halpern}.
\newblock The resource theory of informational nonequilibrium in
  thermodynamics.
\newblock \emph{Physics Reports}, 583:\penalty0 1--58, 2015.
\newblock ISSN 0370-1573.
\newblock \doi{https://doi.org/10.1016/j.physrep.2015.04.003}.

\bibitem[Streltsov et~al.(2018)Streltsov, Kampermann, Wölk, Gessner, and
  Bruß]{streltsov_et_al_2018}
Alexander Streltsov, Hermann Kampermann, Sabine Wölk, Manuel Gessner, and
  Dagmar Bruß.
\newblock Maximal coherence and the resource theory of purity.
\newblock \emph{New Journal of Physics}, 20\penalty0 (5):\penalty0 053058, may
  2018.
\newblock \doi{10.1088/1367-2630/aac484}.

\bibitem[Manzano(2020)]{manzano_2020}
Daniel Manzano.
\newblock A short introduction to the lindblad master equation.
\newblock \emph{{AIP} Advances}, 10\penalty0 (2):\penalty0 025106, Feb 2020.
\newblock \doi{10.1063/1.5115323}.

\bibitem[Verstraete et~al.(2004)Verstraete, Garc\'{\i}a-Ripoll, and
  Cirac]{verstraete_et_al_2004}
F.~Verstraete, J.~J. Garc\'{\i}a-Ripoll, and J.~I. Cirac.
\newblock Matrix product density operators: Simulation of finite-temperature
  and dissipative systems.
\newblock \emph{Phys. Rev. Lett.}, 93:\penalty0 207204, Nov 2004.
\newblock \doi{10.1103/PhysRevLett.93.207204}.

\bibitem[Fishman et~al.(2022)Fishman, White, and
  Stoudenmire]{fishman_et_al_2021}
Matthew Fishman, Steven~R. White, and E.~Miles Stoudenmire.
\newblock The itensor software library for tensor network calculations.
\newblock \emph{SciPost Phys. Codebases}, page~4, 2022.
\newblock \doi{10.21468/SciPostPhysCodeb.4}.

\bibitem[Gangat et~al.(2017)Gangat, I, and Kao]{gangat_et_al_2017}
Adil~A. Gangat, Te~I, and Ying-Jer Kao.
\newblock Steady states of infinite-size dissipative quantum chains via
  imaginary time evolution.
\newblock \emph{Phys. Rev. Lett.}, 119:\penalty0 010501, Jul 2017.
\newblock \doi{10.1103/PhysRevLett.119.010501}.

\bibitem[Fischer et~al.(2016)Fischer, Maksymenko, and
  Altman]{fischer_et_al_2016}
Mark~H Fischer, Mykola Maksymenko, and Ehud Altman.
\newblock Dynamics of a many-body-localized system coupled to a bath.
\newblock \emph{Phys. Rev. Lett.}, 116:\penalty0 160401, Apr 2016.
\newblock \doi{10.1103/PhysRevLett.116.160401}.

\bibitem[van Nieuwenburg et~al.(2017)van Nieuwenburg, Malo, Daley, and
  Fischer]{van_nieuwenburg_et_al_2017}
EPL van Nieuwenburg, J~Yago Malo, AJ~Daley, and MH~Fischer.
\newblock Dynamics of many-body localization in the presence of particle loss.
\newblock \emph{Quantum Science and Technology}, 3\penalty0 (1):\penalty0
  01LT02, dec 2017.
\newblock \doi{10.1088/2058-9565/aa9a02}.

\bibitem[Lenar\ifmmode \check{c}\else \v{c}\fi{}i\ifmmode~\check{c}\else
  \v{c}\fi{} et~al.(2020)Lenar\ifmmode \check{c}\else
  \v{c}\fi{}i\ifmmode~\check{c}\else \v{c}\fi{}, Alberton, Rosch, and
  Altman]{lenarcic_et_al_2020}
Zala Lenar\ifmmode \check{c}\else \v{c}\fi{}i\ifmmode~\check{c}\else
  \v{c}\fi{}, Ori Alberton, Achim Rosch, and Ehud Altman.
\newblock Critical behavior near the many-body localization transition in
  driven open systems.
\newblock \emph{Phys. Rev. Lett.}, 125:\penalty0 116601, Sep 2020.
\newblock \doi{10.1103/PhysRevLett.125.116601}.

\bibitem[White et~al.(2018)White, Zaletel, Mong, and Refael]{white_et_al_2018}
Christopher~David White, Michael Zaletel, Roger S.~K. Mong, and Gil Refael.
\newblock Quantum dynamics of thermalizing systems.
\newblock \emph{Phys. Rev. B}, 97:\penalty0 035127, Jan 2018.
\newblock \doi{10.1103/PhysRevB.97.035127}.

\bibitem[Jaschke et~al.(2018)Jaschke, Montangero, and Carr]{jaschke_et_al_2018}
Daniel Jaschke, Simone Montangero, and Lincoln~D Carr.
\newblock One-dimensional many-body entangled open quantum systems with tensor
  network methods.
\newblock \emph{Quantum Science and Technology}, 4\penalty0 (1):\penalty0
  013001, nov 2018.
\newblock \doi{10.1088/2058-9565/aae724}.

\bibitem[Dupont et~al.(2021)Dupont, Sherman, and Moore]{dupont_et_al_2021}
Maxime Dupont, Nicholas~E. Sherman, and Joel~E. Moore.
\newblock Spatiotemporal crossover between low- and high-temperature dynamical
  regimes in the quantum heisenberg magnet.
\newblock \emph{Phys. Rev. Lett.}, 127:\penalty0 107201, Aug 2021.
\newblock \doi{10.1103/PhysRevLett.127.107201}.

\bibitem[Lerose et~al.(2021)Lerose, Sonner, and Abanin]{lerose_et_al_2021}
Alessio Lerose, Michael Sonner, and Dmitry~A. Abanin.
\newblock Influence matrix approach to many-body floquet dynamics.
\newblock \emph{Phys. Rev. X}, 11:\penalty0 021040, May 2021.
\newblock \doi{10.1103/PhysRevX.11.021040}.

\bibitem[Sonner et~al.(2021)Sonner, Lerose, and Abanin]{sonner_et_al_2021}
Michael Sonner, Alessio Lerose, and Dmitry~A. Abanin.
\newblock Influence functional of many-body systems: Temporal entanglement and
  matrix-product state representation.
\newblock \emph{Annals of Physics}, 435:\penalty0 168677, 2021.
\newblock ISSN 0003-4916.
\newblock \doi{https://doi.org/10.1016/j.aop.2021.168677}.

\bibitem[Lukin et~al.(2019)Lukin, Rispoli, Schittko, Tai, Kaufman, Choi,
  Khemani, Léonard, and Greiner]{lukin_et_al_2019}
Alexander Lukin, Matthew Rispoli, Robert Schittko, M.~Eric Tai, Adam~M.
  Kaufman, Soonwon Choi, Vedika Khemani, Julian Léonard, and Markus Greiner.
\newblock Probing entanglement in a many-body-localized system.
\newblock \emph{Science}, 364\penalty0 (6437):\penalty0 256--260, 2019.
\newblock \doi{10.1126/science.aau0818}.

\bibitem[Brydges et~al.(2019)Brydges, Elben, Jurcevic, Vermersch, Maier,
  Lanyon, Zoller, Blatt, and Roos]{brydges_et_al_2019}
Tiff Brydges, Andreas Elben, Petar Jurcevic, Benoît Vermersch, Christine
  Maier, Ben~P. Lanyon, Peter Zoller, Rainer Blatt, and Christian~F. Roos.
\newblock Probing r\'enyi entanglement entropy via randomized measurements.
\newblock \emph{Science}, 364\penalty0 (6437):\penalty0 260--263, 2019.
\newblock \doi{10.1126/science.aau4963}.

\bibitem[Preskill(2018)]{preskill_2018}
John Preskill.
\newblock Quantum {C}omputing in the {NISQ} era and beyond.
\newblock \emph{{Quantum}}, 2:\penalty0 79, August 2018.
\newblock ISSN 2521-327X.
\newblock \doi{10.22331/q-2018-08-06-79}.

\bibitem[Brooks(2019)]{brooks_2019}
Michael Brooks.
\newblock Beyond quantum supremacy: the hunt for useful quantum computers.
\newblock \emph{Nature}, 574\penalty0 (7776):\penalty0 19–21, Oct 2019.
\newblock \doi{10.1038/d41586-019-02936-3}.

\bibitem[Carlen(2010)]{carlen_2010}
Eric Carlen.
\newblock Trace inequalities and quantum entropy: an introductory course.
\newblock \emph{Contemp. Math.}, 529:\penalty0 73--140, 2010.
\newblock \doi{10.1090/conm/529/10428}.

\bibitem[Davis(1957)]{davis_1957}
Chandler Davis.
\newblock A schwarz inequality for convex operator functions.
\newblock \emph{Proceedings of the American Mathematical Society}, 8\penalty0
  (1):\penalty0 42--44, 1957.
\newblock ISSN 00029939, 10886826.
\newblock \doi{https://doi.org/10.2307/2032808}.

\end{thebibliography}

\onecolumn\newpage
\appendix

\section{Properties of the configuration coherence}
\label{app: configuration entanglement}
In the main text, we introduce the configuration coherence in terms of the density matrix~\eqref{eq: density matrix}  and the $\C$-matrix~\eqref{eq: configurational matrix} as
\begin{align}
    \label{eq: configuration entanglement (app)}
    \E (\rho) &\coloneqq \sum_{n\neq n'}\tr\left(\C_{n, n'}\right) =
    \sum_{n\neq n'}\sum_{\substack{i_n,j_{n'} \\ \mu_n,\nu_{n'}}}|\rho_{i_n,\mu_n;j_{n'},\nu_{n'}}|^2\,.
\end{align}
Here, we prove the properties~\ref{property: vanishes for separable},\ref{property: convex},\ref{property: unitary invariance}, and~\ref{property: montonous under unital} of the configuration coherence introduced in the main text.
To this end, we provide two alternative definitions of the configuration coherence in terms of local particle number projectors.
As in the main text, we will assume that the system has a fixed particle number $N$, but the discussion is valid for any fixed charge. 

We define the local projectors $\Pi_n^B$ that measure the particle number in subsystem $B$.
The completeness relation of the projectors is $\sum_n \Pi_n^B=\1^B$ and they are orthogonal, $\Pi_n^B \Pi_m^B= \delta_{nm}\Pi_n^B$.
These projectors allow for an alternative formulation of the configuration coherence:
\begin{proposition}
The configuration coherence is given as 
\begin{align}
\label{eq: alternative version of configuration entanglement}
\E (\rho) = \sum_{n\neq n'} ||(\1^A \otimes \Pi_n^B) \rho (\1^A \otimes \Pi_{n'}^B)||^2, 
\end{align}
where $||A||=\sqrt{\tr (A^\dag A)}$ is the Frobenius norm.
\end{proposition}
\begin{proof}

Using the particle number basis $\state{i_n,\mu_n}=\state{i_n}_A\otimes\state{\mu_n}_B$, we can write $\Pi_n^B=\sum_{\mu_n}\state{\mu_n}\conjstate{\mu_n}$.
Expressing the density matrix in the same basis, we find
\begin{align}
    \rho_{n,n'} &\coloneqq (\1^A \otimes \Pi_n^B) \rho (\1^A \otimes \Pi_{n'}^B) \nonumber \\
    &= (\1^A \otimes \sum_{\mu_n}\state{\mu_n}\conjstate{\mu_n}) \left[\sum_{m,m'}\sum_{\substack{i_m,\mu_m\\ j_{m'},\nu_{m'}}} \rho_{i_m,\mu_m;j_{m'},\nu_{m'}} \state{i_m,\mu_m}\conjstate{j_{m'},\nu_{m'}} \right](\1^A \otimes \sum_{\mu_{n'}}\state{\mu_{n'}}\conjstate{\mu_{n'}}) \\
    &= \sum_{\substack{i_n,\mu_n\\ j_{n'},\nu_{n'}}} \rho_{i_n,\mu_n;j_{n'},\nu_{n'}} \state{i_n,\mu_n}\conjstate{j_{n'},\nu_{n'}} \nonumber\,.
\end{align}
Next, we calculate the Frobenius norm squared of this matrix as
\begin{align}
    ||\rho_{n,n'}||^2 = \tr(\rho_{n,n'}^\dag \rho_{n,n'})=\sum_{\substack{i_n,\mu_n\\ j_{n'},\nu_{n'}}} |\rho_{i_n,\mu_n;j_{n'},\nu_{n'}}|^2 = \tr(\C_{n,n'}).
\end{align}
Finally, we find
\begin{align}
    \E(\rho) = \sum_{n\neq n'} \tr(\C_{n,n'})=\sum_{n\neq n'}||\rho_{n,n'}||^2=\sum_{n\neq n'}||(\1^A \otimes \Pi_n^B) \rho (\1^A \otimes \Pi_{n'}^B)||^2\,.
\end{align}
\end{proof}

\begin{proposition}
\label{prop: second alternative form of configuration coherence}
The configuration coherence is given as
\begin{equation}
\label{eq: configuration entanglement as relative purity}
    \E(\rho) = \tr\left((\rho - \rho_\Pi)^2\right)\,,
\end{equation}
with the locally measured density matrix 
\begin{equation}
\label{app: locally measured density matrix}
   \rho_\Pi \coloneqq \sum_n (\1^A \otimes \Pi_n^B) \rho (\1^A \otimes \Pi_n^B)\,.
\end{equation}
\end{proposition}
\begin{proof}
Due to orthogonality of the local projectors, it holds that $||\rho_{n,n'} + \rho_{m,m'}||^2=||\rho_{n,n'}||^2+ ||\rho_{m,m'}||^2 $ for $(n,n')\neq (m,m')$.
By inserting two identities $\1 = \sum_n \1^A\otimes \Pi_n^B$, we find
\begin{align}
\tr\left((\rho - \rho_\Pi)^2\right) &= ||\rho - \rho_\Pi||^2 =||\1 \rho \1 - \rho_\Pi||^2 = ||\sum_{n,n'} (\1^A\otimes \Pi_n^B)\rho (\1^A\otimes \Pi_{n'}^B) - \sum_n (\1^A\otimes \Pi_n^B)\rho (\1^A\otimes \Pi_{n}^B)||^2 \nonumber \\ 
&= ||\sum_{n\neq n'} (\1^A\otimes \Pi_n^B)\rho (\1^A\otimes \Pi_{n'}^B)||^2 = ||\sum_{n\neq n'}\rho_{n,n'}||^2= \sum_{n\neq n'}||\rho_{n,n'}||^2 = \E (\rho)\,.
\end{align}
\end{proof}

Note that the form~\eqref{eq: configuration entanglement as relative purity} shows the close relation of the configuration coherence with the number entanglement, which is the relative entropy between $\rho$
 and $\rho_\Pi$~\cite{ma_et_al_2022}. This means that both measures are sensitive to the same flavour of quantum correlations, namely the one between different sub-charge sectors.
 
\begin{proposition}
\label{prop: convexity}
The configuration coherence is convex, i.e., 
\begin{equation}
    \sum_i p_i \E(\rho_i) \geq \E \left(\sum_i p_i \rho_i\right)\, ,
\end{equation}
for density matrices $\rho_i$ and $\sum_i p_i=1$.

\end{proposition}
\begin{proof}
First we prove that 
\begin{equation}
\label{eq: convexity for two operators}
    \lambda \E(\rho) + (1-\lambda) \E(\sigma) \geq \E(\lambda \rho + (1-\lambda) \sigma)\,,
\end{equation}
with density matrices $\rho, \ \sigma$ and $0\leq \lambda \leq 1$.
The measurement $\rho \mapsto \rho_\Pi$ is linear, thus
\begin{align}
    \E (\lambda \rho + (1-\lambda) \sigma)&=\tr (\lambda (\rho - \rho_\Pi) + (1-\lambda) (\sigma - \sigma_\Pi))^2 \nonumber \\
    &\leq \lambda \tr ( (\rho - \rho_\Pi))^2 + (1-\lambda) \tr ((\sigma - \sigma_\Pi))^2 = \lambda \E (\rho) + (1-\lambda) \E (\sigma) \, ,
\end{align}
where the inequality follows from the convexity conservation of the trace function~\cite{carlen_2010} and the convexity of $f(t) = t^2$.

Next, we repeatedly use statement~\eqref{eq: convexity for two operators} and find
\begin{align}
    \E \left(\sum_i p_i \rho_i\right)&=\E \bigg(p_1 \rho_1 + (1-p_1)\underbrace{\sum_{i>1} \frac{p_i}{1-p_1} \rho_i}_{\sigma_1}\bigg) \leq p_1 \E(\rho_1) + (1-p_1)\E(\sigma_1) \nonumber \\ &= p_1 \E(\rho_1) + (1-p_1)\E\bigg(\frac{p_2}{1-p_1} \rho_2 + \frac{1-p_1-p_2}{1-p_1}\underbrace{\sum_{i>2}\frac{p_i}{1-p_1-p_2}\rho_i}_{\sigma_2}\bigg)  \\
    &\leq p_1 \E (\rho_1) + p_2 \E(\rho_2) + (1-p_1-p_2)\E(\sigma_2) \leq \dots \leq \sum_i p_i \E (\rho_i) \, \nonumber .
\end{align}
The $\sigma_j=\sum_{i>j}\frac{p_i}{1-\sum_{k\leq j}p_k} \rho_i$ define valid density matrices because
\begin{align}
    \frac{1}{1-\sum_{k\leq j}p_k}\sum_{i>j}p_i \underbrace{=}_{\sum_i p_i = 1} \frac{1}{1-\sum_{k\leq j}p_k}\left( 1-\sum_{k\leq j}p_k\right) = 1\, .
\end{align}
\end{proof}

\begin{proposition}
\label{prop: vanishes for separable states}
The configuration coherence vanishes for separable states $\rho_{\rm sep}=\sum_i p_i \rho_i^A \otimes \rho_i^B$ with fixed particle number.
\end{proposition}
\begin{proof}
Due to the convexity property~\ref{prop: convexity}, it suffices to show that the configuration coherence vanishes for a product state $\rho = \rho^A \otimes \rho^B$.
We use the total particle number operator
\begin{equation}
\hat{N}=\hat{N}^A\otimes \1^B + \1^A \otimes \hat{N}^B\, ,
\end{equation}
with the local particle number operators $\hat{N}^{A, B} = \sum_n n \Pi_n^{A,B}$.
Due to the fixed particle number, it holds that $[\rho, \hat{N}]=0$.
Using the partial trace $\tr_A$ over subsystem $A$, we find
\begin{align}
\label{app: proof of vanishing for product state}
    0&=\tr_A\left([\rho, \hat{N}]\right)=\tr_A \left([\rho^A \otimes \rho^B, \hat{N}^A\otimes \1^B + \1^A \otimes \hat{N}^B]\right) \nonumber \\
    &=\tr\left(\rho^A \hat{N^A}\right) \rho^B + \underbrace{\tr(\rho^A)}_{=1} \left(\rho^B \hat{N}^B\right) 
    - \underbrace{\tr\left(\hat{N^A}\rho^A \right)}_{=\tr\left(\rho^A \hat{N^A}\right)} \rho^B - \underbrace{\tr(\rho^A)}_{=1} \left(\hat{N}^B\rho^B \right) \\
    &= \rho^B\hat{N}^B - \hat{N}^B\rho^B = [\rho^B, \hat{N}^B] = \sum_n n [\rho^B, \Pi_n^B]\nonumber\, .
\end{align}
Due to the orthogonality of the projectors, it follows $[\rho^B, \Pi_n^B]=0, \ \forall n$.
Thus, for a product state, the locally measured density matrix~\eqref{app: locally measured density matrix} is equivalent to the density matrix, $(\rho^A \otimes \rho^B)_\Pi=\rho^A \otimes \rho^B$.
From the form~\eqref{eq: configuration entanglement as relative purity} it follows that the configuration coherence must vanish for a product state, and therefore for separable states.


\end{proof}

\begin{proposition}
\label{prop: invariant under local unitary}
    The configuration coherence is invariant under local particle number conserving unitary operations.
\end{proposition}
\begin{proof}
    Let $U=U^A \otimes U^B$ be a local unitary transformation, i.e., $U U^\dag = U^\dag U=\1$. The particle number conservation translates to $[U, \hat{N}]=0$. By replacing $\rho \mapsto U$ in~\eqref{app: proof of vanishing for product state} and noting that $\tr_A U^A \neq 0$, we find $[U, \1^A \otimes \Pi_n^B]=0$ $\forall n$. From the latter, it follows for the locally measured density matrix~\eqref{app: locally measured density matrix} that $U\rho_\Pi U^\dag = (U\rho U^\dag)_\Pi$. Thus, 
    \begin{align}
        \E (U\rho U^\dag) &= \tr \left( \left(U\rho U^\dag - (U \rho U^\dag)_\Pi\right)^2\right) = \tr \left( \left(U\rho U^\dag - U \rho_\Pi U^\dag \right)^2\right) \nonumber \\ 
        &= \tr \left( \left(U(\rho - \rho_\Pi)U^\dag\right)^2 \right) = \tr \left((\rho - \rho_\Pi)^2\right) = \E (\rho)\,,
    \end{align}
    where the second-to-last equality follows from the unitarity of $U$ and the cyclic property of the trace.
\end{proof}

\begin{proposition}
The configuration coherence is monotonous under local purity non-increasing particle number conserving operations.
\end{proposition}
\begin{proof}
We consider a local charge conserving operation $\rho \mapsto \mathcal{E} (\rho) = \sum_i K_i \rho K_i^\dag$, with Kraus operators $K_i = K_i^A \otimes K_i^B$ satisfying $\sum_i K_i^\dag K_i = \1$. The particle number conservation implies $[K_i, \1^A \otimes \Pi_n^B]=0$ (proof is equivalent to proof of $[U, \1^A \otimes \Pi_n^B]=0$ for~\ref{prop: invariant under local unitary}). The most general purity non-increasing maps are unital maps~\cite{gour_et_al_2015, streltsov_et_al_2018}, i.e., $\sum_i K_i K_i^\dag = \1$. 
Using the fact that $||\rho_{n,n'} + \rho_{m,m'}||^2 = ||\rho_{n,n'}||^2 + ||\rho_{m,m'}||^2$ for $(n,n')\neq (m,m')$, we rewrite the configuration coherence~\eqref{eq: configuration entanglement (app)} as
\begin{align}
    \E(\rho) = \sum_{n'\neq n} ||\rho_{n,n'}||^2 = \sum_{n'>n} ||\rho_{n,n'}||^2 + ||\rho_{n',n}||^2 = \sum_{n'>n} ||\rho_{n,n'} + \rho_{n',n}||^2 =  \sum_{n'>n} \tr (\rho_{n,n'} + \rho_{n',n})^2\,.
\end{align}
Note that the matrices $ \rho_{n,n'} + \rho_{n',n}$ are Hermitian.

We prove the monotonicity in two steps.
First, we show that $\mathcal{E}(\rho)_{n,n'} + \mathcal{E}(\rho)_{n',n}=\mathcal{E}( \rho_{n,n'} + \rho_{n',n})$.
Secondly, we prove that $||\mathcal{E}( \rho_{n,n'} + \rho_{n',n})||^2\leq || \rho_{n,n'} + \rho_{n',n}||^2$ $\forall n'> n$ using Jensen's trace inequality.

For the first step, we have
\begin{align}
    \mathcal{E}(\rho)_{n,n'} &= (\1^A \otimes \Pi_n^B) \mathcal{E}(\rho) (\1^A \otimes \Pi_{n'}^B) = (\1^A \otimes \Pi_n^B) \sum_i K_i \rho K_i^\dag (\1^A \otimes \Pi_{n'}^B) \nonumber \\ 
    &\underbrace{=}_{[K_i, \1^A\otimes\Pi_n^B]=0} \sum_i K_i (\1^A \otimes \Pi_n^B) \rho (\1^A \otimes \Pi_{n'}^B) K_i^\dag = \mathcal{E}(\rho_{n,n'})\,,
\end{align}
From the linearity of $\mathcal{E}$, it follows that
\begin{align}
    \mathcal{E}(\rho)_{n,n'} + \mathcal{E}(\rho)_{n',n}=\mathcal{E}(\rho_{n,n'}) + \mathcal{E}(\rho_{n',n})=\mathcal{E}( \rho_{n,n'} + \rho_{n',n}) \, .
\end{align}

For the second step, we use Jensen's trace inequality~\cite{davis_1957}, which states that
\begin{align}
\tr \left(f\left( \sum_i K_i X_i K_i^\dag \right) \right) \leq \tr \left( \sum_i K_i f\left(X_i\right) K_i^\dag\right)\,,
\end{align}
for convex $f$, Hermitian matrices $X_i$ and quadratic $K_i$ satisfying $\sum_i K_i K_i^\dag=\1$. From setting $X_i=\rho_{n,n'}+\rho_{n',n} \ \forall i$ and using the unitality of $\mathcal{E}$ and the convexity of $f(t)=t^2$, it follows that
\begin{align}
    \tr (\mathcal{E}(\rho)_{n,n'} + \mathcal{E}(\rho)_{n',n})^2 &= \tr (\mathcal{E}(\rho_{n,n'} + \rho_{n',n}))^2=\tr \left(\sum_i K_i (\rho_{n,n'} + \rho_{n',n}) K_i^\dag \right)^2 
    \nonumber \\
    &\underbrace{\leq}_{\rm Jensen} \tr \left(\sum_i K_i (\rho_{n,n'} + \rho_{n',n})^2 K_i^\dag \right)
    =\tr \bigg((\rho_{n,n'} + \rho_{n',n})^2 \underbrace{\sum_i K_i^\dag K_i}_{=\1} \bigg)  \nonumber \\ &=
    \tr (\rho_{n,n'} + \rho_{n',n})^2 \, .
\end{align}

Finally, we find that
\begin{align}
    \E (\mathcal{E}(\rho))=\sum_{n'>n} \tr (\mathcal{E}(\rho)_{n,n'} + \mathcal{E}(\rho)_{n',n})^2 \leq \sum_{n'>n} \tr (\rho_{n,n'} + \rho_{n',n})^2 = \E (\rho)
\end{align}
Note that for general quantum maps, the $K_i$ do not have to be quadratic, and in that case Jensen's trace inequality does not apply.
However, we are interested in situations where the system's Hilbert space is fixed, which requires that the map $\mathcal{E}$ is built from quadratic Kraus operators $K_i$.
\end{proof}

\section{Rank of the $\C$-matrix}
\label{app: rank of C-matrix}
	In the main text, we introduce the $\C$-matrix~\eqref{eq: configurational matrix} whose eigenvalues define the operator-space entanglement spectrum (OSES). For fixed particle number, the matrix has a block diagonal structure [cf.~eq.~\eqref{eq: configurational matrix for n particles} and Fig.~\ref{fig: Systems and Failure of OSES}(d)]. Here, we discuss the rank of these blocks in more detail.
	
	In a matrix product density operator (MPDO) representation of the mixed state density matrix~\cite{verstraete_et_al_2004}, the rank of the $\C$-matrix defines the necessary bond dimension $\chi_{\rm exact}$ of the MPDO to represent the state exactly. We can calculate the rank by summing over the ranks of the individual blocks of the $\C$-matrix. 
	
	The blocks are given by
	\begin{equation}
	\C_{n,n'}=\sum_{\substack{i_{n},j_{n'}  k_n,l_{n'}}} c^{\phantom{\dagger}}_{i_n,j_{n'};k_n,l_{n'}} \supervec{i_n,j_{n'}}\conjsupervec{k_n,l_{n'}}\,,
	\label{eq: C matrix block}
	\end{equation}
with coefficients
\begin{equation}
c^{\phantom{\dagger}}_{i_n,j_{n'};k_n,l_{n'}} = \sum_{\mu_n,\nu_{n'}}\rho_{i_n,\mu_{n};j_{n'},\nu_{n'}} \rho_{k_n,\mu_n;l_{n'},\nu_{n'}}^*\,.
\end{equation}
Reordering the summation, we can write the block~\eqref{eq: C matrix block} as
\begin{equation}
    \C_{n,n'} = \sum_{\mu_n,\nu_{n'}} v_{\mu_n,\nu_{n'}} v_{\mu_n,\nu_{n'}}^\dag\,, 
    \label{eq: C matrix block as sum over rank 1 matrices}
\end{equation}
with vectors
\begin{equation}
    v_{\mu_n,\nu_{n'}} = \sum_{i_{n},j_{n'}} \rho_{i_n,\mu_{n};j_{n'},\nu_{n'}} \supervec{i_n,j_{n'}}\,.
    \label{eq: vector}
\end{equation}
The length of the vector~\eqref{eq: vector} and thus the size of the block~\eqref{eq: C matrix block} is determined by the number of possible supervectors $\supervec{i_n, j_{n'}}$.

If we denote the size of subsystem $A$ by $L_A$, then there are $\binom{L_A}{n}$ possible states $\state{i_n}$ of $n$ particles in subsystem $A$. It follows that the size of the block~\eqref{eq: C matrix block} is given by $\binom{L_A}{n}\times \binom{L_A}{n'}$. At the same time, the form~\eqref{eq: C matrix block as sum over rank 1 matrices} reveals that the block $\C_{n, n'}$ is a sum over rank-1 matrices $v_{\mu_n,\nu_{n'}} v_{\mu_n,\nu_{n'}}^\dag$.
Consequently, there are $\binom{L_B}{N-n} \times \binom{L_B}{N-n'}$ of these matrices, with $L_B$ the size of subsystem $B$. 
As the rank of a matrix sum is bounded by the sum of the ranks of the summands, we find that the block $\mathcal{C}_{n, n'}$ has a maximal rank 
\begin{equation}
\label{eq: single block rank}
    {\rm rank}\mathcal{C}_{n, n'}={\rm min}\left( \binom{L_A}{n}\times \binom{L_A}{n'},\binom{L_B}{N-n}\times\binom{L_B}{N-n'}\right)\,,
\end{equation}
with $L_A$ ($L_B$) the size of subsystem A (B).
It follows that the maximal bond dimension for $N$ particles is given by
\begin{equation}
    \chi_{\rm exact} = \sum_{n,n'=0}^N {\rm min}\left( \binom{L_A}{n} \times \binom{L_A}{n'},\binom{L_B}{N-n}\times\binom{L_B}{N-n'}\right)\,.
    \label{eq: max rank}
\end{equation}

The maximal rank~\eqref{eq: max rank} for $N$ particles in a system of size $L$ scales as $\chi_{\rm exact} \propto L^N$ (stemming from the blocks with $n+n'=N$). Thus, if the number of particles is fixed, the maximal rank has a power-law scaling with system size, as opposed to the exponential scaling in the general setting. This permits simulation of very large system sizes with fixed particle number. Table~\ref{table: ranks} shows the maximal ranks per block for the example of $N=2$ particles. As expected, the maximal bond dimension scales as $L^2$.

	\begin{table}[h!]
\centering
\begin{tabular}{c c c c}
\hline
\hline
block & $(n,n')$ & degeneracy & maximal rank \\
\hline
AA & $(2,2)$ &  1 & 1 \\

AC & $(2,1)$ &  2 & $\frac{L}{2}$ \\

AB & $(1,1)$ &  1 & $\left(\frac{L}{2}\right)^2$ \\

BC & $(0,1)$ &  2 & $\frac{L}{2}$ \\

BB & $(0,0)$ &  1 & 1 \\

CC & $(2,0)$ & 2 & $\frac{L}{4}\left(\frac{L}{2} - 1\right)$ \\
\hline
\multicolumn{3}{c}{maximal bond dimension } & $\frac{L^2}{2}+\frac{3L}{2}+2$\\
\hline
\end{tabular}
\caption{Maximal ranks of the blocks of the $\mathcal{C}$-matrix for $N=2$ particles on a chain of length $L$ and bipartition in the middle. The maximal bond dimension is obtained as the sum over the maximal ranks times the corresponding degeneracies.}
\label{table: ranks}
\end{table}

Several factors can reduce the maximal bond dimension~\eqref{eq: max rank}.
In particular, under the assumption of local decoherence, cf.~Ref.~\cite{nieuwenburg_zilberberg_2018}, the blocks $\C_{N,0}$ and $\C_{0, N}$ become rank 1.

\section{Configuration coherence and negativity for a single particle}
In this section, we show the equivalence of the configuration coherence~\eqref{eq: configuration entanglement} to the negativity for a mixed state of a single particle.
Recall that the negativity is given as the sum over the negative eigenvalues of the partially transposed density matrix~\cite{vidal_werner_2002}.

We use the particle number basis $\state{i_n,\mu_n}=\state{i_n}_A\otimes\state{\mu_n}_B$ and that for a single particle, it holds $n\in[0,1]$.
Therefore, we can define $\state{i}_A\coloneqq \state{i_1}_A$ and $\state{\mu}_B\coloneqq \state{\mu_1}_B$. 
Note that the only state with zero particles in a subsystem is the vacuum state $\state{0}_{A,B}$.
With this, the density matrix~\eqref{eq: density matrix} for a single particle can be written as
\begin{align}
	\label{eq: density matrix single particle}
	\rho &= \left(\sum_{i,j \in A} \rho_{i,0; j,0} \state{i}\conjstate{j}_A\right) \otimes \pure{0}_B \nonumber 
	+ \pure{0}_A \otimes \left(\sum_{\mu,\nu \in B} \rho_{0,\mu; 0,\nu} \state{\mu}\conjstate{\nu}_B\right) \\
	&+ \sum_{i\in A, \mu \in B} \rho_{i,0; 0, \mu} \state{i}_A\state{0}_B\conjstate{0}_A\conjstate{\mu}_B + \rho_{0, \mu; i, 0} \state{0}_A\state{\mu}_B\conjstate{i}_A\conjstate{0}_B\,.
\end{align}

The partial transpose of~\eqref{eq: density matrix single particle} has a block-diagonal form, and only one block has a negative eigenvalue. It is the block describing the particle in a coherent cross-boundary state,
\begin{align}
	\label{eq: partial transpose coherence block}
	\rho^{T_A}_C &= \sum_{i\in A, \mu \in B} \rho_{i,0; 0, \mu} \state{0}_A\state{0}_B\conjstate{i}_A\conjstate{\mu}_B 
	 + \rho_{i,0; 0, \mu}^* \state{i}_A\state{\mu}_B\conjstate{0}_A\conjstate{0}_B 
	= \state{0}\conjstate{\phi}_{AB} + \state{\phi}\conjstate{0}_{AB}\,, 
\end{align}
with the non normalized vector
\begin{equation}
	\state{\phi}_{AB} \coloneqq \sum_{i\in A, \mu \in B} \rho_{i,0; 0, \mu}^* \state{i}_A\state{\mu}_B\,.
\end{equation}
A rank-2 block of the form~\eqref{eq: partial transpose coherence block} has two eigenvalues $\pm \sqrt{\langle 0|0\rangle_{AB} \langle \phi|\phi\rangle_{AB}}=\pm \sqrt{\langle \phi|\phi\rangle_{AB}}$. Thus, the negativity for the single particle state~\eqref{eq: density matrix single particle} is given by
\begin{equation}
	\label{eq: negativity of single particle state}
	\N(\rho) = \sqrt{\langle \phi|\phi\rangle_{AB}}=\sqrt{\sum_{i\in A, \mu \in B} |\rho_{i,0; 0, \mu}|^2}\,.
\end{equation}

For the configuration coherence~\eqref{eq: configuration entanglement} of the single particle state~\eqref{eq: density matrix single particle}, we find
\begin{equation}
	\label{eq: trace over coherence block of configurational matrix}
	\E (\rho) = \tr \left( \C_{0, 1} \right) + \tr \left( \C_{1, 0} \right)=2\times \tr \left(\C_{0, 1}\right) = 2 \sum_{i\in A, \mu \in B} |\rho_{i,0;0, \mu}|^2 = 2\N(\rho)^2\,.
\end{equation}
Therefore, for mixed states of a single particle, the configuration coherence equals twice the negativity squared.

\end{document}